\documentclass[conference]{IEEEtran}
\IEEEoverridecommandlockouts
\usepackage[ascii]{inputenc}
\usepackage[dvipsnames]{xcolor}
\usepackage[bookmarksnumbered,linktocpage,hypertexnames=false,colorlinks=true,linkcolor=NavyBlue,urlcolor=NavyBlue,citecolor=ForestGreen,anchorcolor=green,breaklinks=true,pagebackref=true]{hyperref}
\usepackage{xurl}
\usepackage[normalem]{ulem}
\usepackage{bbm,braket,microtype,mathrsfs,amsmath,amssymb,xcolor,amsthm,amsfonts,latexsym,mathtools,graphicx,booktabs,bm,xspace,float,mathdots,ellipsis,mleftright}
\usepackage[T1]{fontenc}
\usepackage{textcomp}
\usepackage[english]{babel}
\usepackage[capitalize,nameinlink]{cleveref}
\usepackage[shortlabels]{enumitem}
\newtheorem{thm}{Theorem}[section]
\newtheorem{lem}[thm]{Lemma}
\newtheorem{prop}[thm]{Proposition}
\newtheorem{cor}[thm]{Corollary}
\theoremstyle{definition}
\newtheorem{defn}[thm]{Definition}
\theoremstyle{remark}
\newtheorem{rem}[thm]{Remark}
\newtheorem{exa}[thm]{Example}
\AddToHook{env/lem/begin}{\crefalias{thm}{lem}}
\AddToHook{env/prop/begin}{\crefalias{thm}{prop}}
\AddToHook{env/cor/begin}{\crefalias{thm}{cor}}
\AddToHook{env/defn/begin}{\crefalias{thm}{defn}}
\AddToHook{env/rem/begin}{\crefalias{thm}{rem}}
\AddToHook{env/exa/begin}{\crefalias{thm}{exa}}

\crefname{thm}{Theorem}{Theorems}
\crefname{lem}{Lemma}{Lemmata}

\numberwithin{equation}{section}
\DeclareMathOperator{\tr}{tr}
\DeclareMathOperator{\conv}{conv}
\DeclareMathOperator{\supp}{supp}
\DeclareMathOperator{\diagmat}{diag}
\DeclareMathOperator{\diagvec}{diag}

\DeclareMathOperator{\GL}{GL}
\DeclareMathOperator{\PD}{PD}
\DeclareMathOperator{\U}{U}

\DeclareMathOperator{\dom}{dom}
\DeclareMathOperator{\Herm}{Herm}
\DeclareMathOperator{\dist}{dist}
\DeclareMathOperator{\spec}{spec}
\DeclareMathOperator{\diag}{diag}
\DeclareMathOperator{\rk}{rk}
\DeclareMathOperator{\ncrk}{ncrk}

\DeclareMathOperator{\Lin}{Lin}
\DeclareMathOperator{\SR}{SR}

\DeclareMathOperator{\Sym}{S}
\DeclareMathOperator{\Alt}{\Lambda}
\newcommand{\diff}{\mathrm{d}}

\newcommand{\PP}{\mathbbm{P}}
\newcommand{\RR}{\mathbbm{R}}
\newcommand{\ZZ}{\mathbbm{Z}}
\newcommand{\CC}{\mathbbm{C}}

\newcommand{\rmH}{\mathbb{H}}
\newcommand{\rmT}{\mathrm{T}}
\newcommand{\VV}{\mathcal{V}}
\newcommand{\MM}{\mathcal{M}}
\newcommand{\HH}{\mathcal{H}}
\newcommand{\BB}{\mathcal{B}}
\newcommand{\op}{\oplus}
\newcommand{\ot}{\otimes}

\DeclarePairedDelimiter\norm{\lVert}{\rVert}

\allowdisplaybreaks[4]
\newenvironment{restatedtheorem}[1]{%
  \begingroup
  \renewcommand{\thethm}{\ref{#1}}%
  \begin{thm}[Restated]%
}{%
  \end{thm}
  \addtocounter{thm}{-1}%
  \endgroup
}

\begin{document}
\title{Strassen's support functionals\\coincide with the quantum functionals\thanks{
This research is supported by the European Union (ERC Grant SYMOPTIC, 101040907), by the Deutsche Forschungsgemeinschaft (DFG, German Research Foundation, 556164098), and by the Deutsche Forschungsgemeinschaft under Germany's Excellence Strategy - EXC-2111 - 390814868.}}
\author{
    \IEEEauthorblockN{Keiya Sakabe}
    \IEEEauthorblockA{
    \textit{LMU Munich} \& \textit{MCQST}\\
    Munich, Germany\\
    keiya.sakabe@lmu.de}
    \and
    \IEEEauthorblockN{M. Levent Do\u{g}an} \IEEEauthorblockA{
    \textit{LMU Munich} \& \textit{MCQST}\\
    Munich, Germany\\
    mahmut.dogan@lmu.de}
    \and
    \IEEEauthorblockN{Michael Walter} \IEEEauthorblockA{
    \textit{LMU Munich} \& \textit{MCQST}\\
    Munich, Germany\\
    \textit{University of Amsterdam} \& \textit{QuSoft}\\
    Amsterdam, The Netherlands\\
    michael.walter@lmu.de}
}

\maketitle

\begin{abstract}
Strassen's asymptotic spectrum offers a framework for analyzing the complexity of tensors.
It has found applications in diverse areas, from computer science to additive combinatorics and quantum information.
A long-standing open problem, dating back to 1991, asks whether Strassen's support functionals are universal spectral points, that is, points in the asymptotic spectrum of tensors.
In this paper, we answer this question in the affirmative for tensors over complex numbers by proving that the support functionals coincide with the quantum functionals---universal spectral points that are defined via entropy optimization on entanglement polytopes.
We obtain this result as a special case of a general minimax formula for convex optimization on entanglement polytopes (and more general moment polytopes).
Our formula can be interpreted as a very general ``classical-quantum correspondence'' analogous to the well-known relation between the Shannon entropy and the von Neumann entropy.
Its proof is based on a recent Fenchel-type duality theorem on Hadamard manifolds due to Hirai.
In addition to settling Strassen's question, our results yield a unified and simpler approach to several other tensor invariants, including the asymptotic slice rank.

In more detail, our main results are as follows:
\begin{itemize}
    \item A \emph{classical-quantum correspondence} for convex optimization over entanglement polytopes (and other moment polytopes), which are implicitly-defined polytopes that appear in a variety of contexts from algebraic complexity to quantum information: we show that minimization over the entanglement polytope is equivalent to minimization over the support polytopes if the latter is maximized over all basis choices.
    \item As an immediate corollary, we prove that \emph{Strassen's support functionals coincide with the quantum functionals}.
    This implies that the former are universal spectral points, resolving a long-standing open problem in algebraic complexity.
    We also provide a novel characterization of the symmetric quantum functionals in the spirit of Strassen's characterization.
    \item We obtain \emph{new geometric proofs} that the quantum functionals are universal spectral points, and that the symmetric quantum functionals are submultiplicative. These proofs avoid the representation-theoretic machinery used in prior work, suggesting that the classical-quantum correspondence could provide a useful tool for constructing new spectral points.
    \item We derive a novel formula for the \emph{asymptotic slice rank} of a tensor as a minimum, over all basis choices, of the \emph{asymptotic vertex cover number} of associated hypergraphs, which also generalizes to the weighted case.
    \item We obtain direct and uniform new proofs that the $G$-stable rank and the non-commutative rank of tensors can be expressed as optimizations over the moment polytope, and a general method for obtaining such characterizations for other tensor invariants.
\end{itemize}
\end{abstract}

\section{Introduction}\label{sec:introduction}
Entropy plays a central role in computer science, mathematics, physics, and many other fields of sciences.
Most famous of all is the \emph{Shannon entropy} of a probability distribution $p \in \RR^n$:
\begin{align*}
    H(p) = - \sum_{\ell = 1}^n p_\ell \log_2 p_\ell,
\end{align*}
which measures the ``uncertainty'' or ``information'' of the distribution.
It originates in statistical physics, going back to the works of Boltzmann and Gibbs, and plays a foundational role in information theory.
Moving from classical to quantum information, probability distributions get replaced by quantum states or density matrices~$\rho$, which are positive semidefinite matrices of trace one. While the entries of a density matrix are complex numbers, their eigenvalues form a probability distribution, and hence we can compute their Shannon entropy.
The result is called the \emph{von Neumann entropy}:
\begin{align*}
    S(\rho) = -\tr \rho \log_2 \rho = H(\spec \rho),
\end{align*}
where $\spec \rho$ denotes the eigenvalues of the density matrix~$\rho$ sorted (for definiteness) in nonincreasing~order.
It plays a similarly distinguished role in quantum information theory, but also appears in other areas of computer science.
Of course, the Shannon and the von Neumann entropy are tightly connected.
One well-known but slightly non-obvious fact is the following:
\begin{align}\label{eq:motivation}
    S(\rho)
= \min_{u \text{ unitary}} H(\diag(u \rho u^\dagger)),
\end{align}
where $\diag(\cdot)$ sends a matrix to its diagonal entries---in other words, the right-hand side optimizes the Shannon entropy of the diagonal entries, which form a probability distribution, with respect to an arbitrary basis.
This basic result has an important interpretation:
in quantum information, classical information can be extracted from quantum states by so-called \emph{measurements} in an arbitrary basis, the outcomes of which are given by the diagonal entries in that basis.
Thus, \cref{eq:motivation} states that the classical entropy is always at least as large as the quantum one---measurements can only increase the ``uncertainty''---and there is a basis choice (the eigenbasis of~$\rho$) such that both quantities coincide.

In this paper, we prove a correspondence between quantum (spectral) and classical (diagonal) quantities that is analogous to~\eqref{eq:motivation} but applies in a vastly more general situation:
instead of matrices, we consider \emph{tensors} (or indeed vectors in an arbitrary Lie group representation), whose spectral information is given by convex \emph{entanglement polytopes} (or indeed general moment polytopes), which arise naturally in diverse areas from quantum information to algebraic complexity and scaling algorithms and will be defined below, and we consider not just the Shannon and von Neumann entropy but general concave (or convex) unitarily-invariant objectives.
Our classical-quantum correspondence allows us to establish a number of new results on tensor invariants---importantly to resolve a long-standing open problem in algebraic complexity going back to Strassen---and yields new succinct proofs using a unified perspective.

The remainder of this introduction is organized as follows.
In \cref{sec:intro-strassen}, we start by discussing the open problem in algebraic complexity that motivated of our work and our resolution.
Next, we present our general classical-quantum correspondence, which takes the form of a minimax formula, first in the important setting of moment polytopes, and then in its most general natural formulation in terms of gradient optimization on Hadamard manifolds (\cref{intro:minimax}).
Finally we give a number of applications to other tensor invariants that have recently been studied intensely (\cref{sec:intro-applications}).
We close with an outlook discussing directions for future research opened up by our result (\cref{intro:outlook}).

\subsection{Strassen's support functionals coincide with the quantum functionals}\label{sec:intro-strassen}
A central motivation for our results comes from a long-standing open question in \emph{algebraic complexity}, and, more broadly, in the \emph{asymptotic theory of tensors} (which itself has broad applications in mathematics, computer science, and physics), which we address in this paper.
Remarkably, entropy has been understood to also play a distinguished role in this setting.
In algebraic complexity theory, tensors capture multilinear problems such as the problem of multiplying $n \times n$ matrices.
Strassen's theory of the \emph{asymptotic spectrum} of tensors establishes an important unifying framework for studying the computational complexity of multilinear problems~\cite{Strassen-86, Strassen-88, Strassen-91, Strassen-05}.
We briefly sketch this theory.
Given a $d$-tensor~$t\in \CC^{n_1}\otimes\dots\otimes\CC^{n_d}$, we say that $t$ \emph{restricts} to a tensor~$s\in\CC^{m_1}\otimes\dots\otimes\CC^{m_d}$, and write~$t\geq s$, if there are linear maps~$A_i \colon \CC^{n_i}\rightarrow \CC^{m_i}$ such that~$(A_1\otimes\dots\otimes A_d)t = s$.
We say that~$t$ \emph{asymptotically restricts} to~$s$ if~$t^{\otimes n + o(n)}\geq s^{\otimes n}$.
When one of the tensors is the \emph{unit tensor}~$\langle r\rangle\coloneqq \sum_{i=1}^r e_i\otimes\dots\otimes e_i$, asymptotic restriction captures the \emph{asymptotic rank and subrank} of a tensor, which directly bounds the algebraic complexity.
In general, it is a difficult problem to determine when two tensors asymptotically restrict each other.
Strassen discovered an important structural simplification: the asymptotic restriction problem is fully captured by certain functionals.
More precisely, for a family $\mathcal{X}$ of $d$-tensors that is closed under tensor product, direct sums, and contains the unit tensors~$\langle r\rangle$, define its asymptotic spectrum as the set of all functions from~$\mathcal{X}$ to~$\RR_{\geq 0}$ that satisfy the following four axioms:
they are
\begin{enumerate}[(a),noitemsep]
    \item monotone under restriction,
    \item multiplicative under tensor product,
    \item additive under direct sum, and
    \item normalized to have value $r$ on the unit tensor $\langle r\rangle$.
\end{enumerate}
Strassen's spectral theorem states that such functionals, called \emph{spectral points}, are invariants that characterize the asymptotic restriction between tensors fully.
Namely, for any two tensors~$s,t\in\mathcal{X}$:
\begin{align*}t^{\otimes n + o(n)} \geq s^{\otimes n}&\\\iff &\ \phi(t)\geq \phi(s) \text{ for every spectral point } \phi
\end{align*}
In particular, spectral points provide upper bounds on the asymptotic subrank and lower bounds on the asymptotic rank of a tensor.
This makes the asymptotic spectrum particularly powerful in computational complexity.
Remarkably, while originally invented in the context of complexity theory, and in particular the long-standing open problem of characterizing the exponent of matrix multiplication, Strassen's theory has found applications in a diverse set of areas such as additive combinatorics~\cite{Alon-13,Tao-16,Tao-Sawin-16,Blasiak-17,Kleinberg-18} and quantum information theory, where tensors describe quantum states and restriction corresponds to entanglement transformations~\cite{dur2000three,walter2013entanglement,zuiddam2018algebraic,christandl2024tensor}.

Spectral points for the family of \emph{all} $d$-tensors, which are called \emph{universal} spectral points, are of central interest.
The flattening ranks---the matrix rank of the flattenings of tensors---serve as straightforward examples.
The construction of nontrivial (that is, different from the flattening ranks) universal spectral points turned out to be a particularly challenging task.
Early on, Strassen~\cite{Strassen-91} introduced a class of candidate spectral points, parameterized by probability distributions~$\theta\in\Theta\coloneqq \{\theta\in\RR_{\geq 0}^d\mid \sum_j \theta_j = 1\}$.
They can be defined as follows:
for a tensor~$t\in\CC^{n_1}\otimes\CC^{n_2}\otimes\dots\otimes\CC^{n_d}$, define its \emph{support polytope}
\begin{align}\label{eq:support polytope intro}
    \Omega(t) &\coloneqq \conv\{(e_{j_1},e_{j_2},\dots,e_{j_d}) \mid t_{j_1,j_2,\dots,j_d}\neq 0\}\\&\subseteq \RR^{n_1} \!\times\! \cdots \!\times\! \RR^{n_d},\notag
\end{align}
where~$e_j$ denotes the standard basis vectors.
Then the \emph{(upper) support functional}~$\zeta^\theta(t)$ for~$\theta \in \Theta$ is%
\footnote{The support functionals can also be defined by considering joint probability distributions~$P$ on~$[n_1] \times \cdots \times [n_d]$ with nonzero entries on the support of~$t$.
Then~$\Omega(t)$ is obtained by computing the tuple~$(p_1,\dots,p_d)$ of marginal distributions of the~$d$ random variables under any such distribution~$P$.}
\begin{equation}\label{eq:zeta}
    \zeta^\theta(t)\coloneqq \min_{g \in \GL}\; \max_{p = (p_1, \ldots, p_d) \in \Omega(g\cdot t)}\; 2^{\sum_{j=1}^d \theta_j H(p_j)},
\end{equation}
where the minimum is over all tuples~$g=(g_1,\dots,g_d)$ in~$\GL\coloneqq \GL(n_1)\times\GL(n_2)\times\dots\times \GL(n_d)$ and~$g \cdot t = (g_1 \ot \cdots \ot g_d) t$ denotes the result of applying the corresponding basis changes on the tensor's indices.
Strassen proved that~$\zeta^{\theta}$ is a spectral point for the family of \emph{oblique} tensors (tensors whose support form an antichain in some basis).
While significant as a spectral point for an important (but non-generic) class of tensors, it left open the question of whether~$\zeta^{\theta}$ is a \emph{universal} spectral point.

In a 2017 breakthrough result, Christandl, Vrana, and Zuiddam constructed the first nontrivial family of universal spectral points~\cite{CVZ2018}.
Their construction, called the \emph{quantum functional}~$F_{\theta}(t)$, is also parameterized by a probability distribution~$\theta \in \Theta$ and admits a similar definition to~$\zeta^{\theta}(t)$, but is in terms of the \emph{entanglement} or \emph{moment polytope} of the tensor, rather than its support polytope.
In more detail, for every~$i\in [d]$, we can view a tensor~$t\in \CC^{n_1}\otimes\dots\otimes\CC^{n_d}$ as a linear map~$t_i : \CC^{n_i}\rightarrow \otimes_{j\neq i} \CC^{n_j}$, called the $i$-th \emph{(principal) flattening} of~$t$.
The \emph{entanglement polytope}~\cite{walter2013entanglement} of a tensor~$t\neq 0$ is defined in terms of the flattenings as
\begin{align}
&\Delta(t) \coloneqq\notag\\
&\left\{\left(\spec \rho_1(s), \spec \rho_2(s),\dots, \spec \rho_d(s) \right)\Big\lvert \; 0\neq s \in \overline{\GL\cdot t} \right\}\notag\\
&\qquad\subseteq \RR^{n_1} \!\times\! \cdots \!\times\! \RR^{n_d},\label{eq:entanglement polytope intro}
\end{align}
where~$\rho_i(s) = \frac{s_i^{\dag}s_i}{\norm{s}^2}$ is the $i$-th \emph{quantum marginal} of~$s$.
Then the \emph{quantum functional} is given by
\begin{equation}\label{eq:F}
    F_\theta(t) \coloneqq \max_{p\in\Delta(t)} \, 2^{\sum_{i=1}^d \theta_i H(p_i)}.
\end{equation}
In contrast to~\eqref{eq:zeta}, no minimization over~$g \in \GL$ is required because it holds that~$\Delta(t)=\Delta(g\cdot t)$ for all~$g\in\GL$, a property that is easy to see from the definition but is not shared by~$\Omega(t)$.%
\footnote{However, the support polytope is invariant under the action of the \emph{maximal torus}~$\rmT$, consisting of the diagonal matrices in~$\GL$. In fact, $\Omega(t)$ can be alternatively defined as the moment polytope of~$t$ corresponding to this action, see \cref{sec:moment-poly}.}
It is a highly nontrivial fact that~$\Delta(t)$ is a convex polytope, as shown in a more general setting by~\cite{Kirwan1984,Ness-Mumford-84} (see also
 \cite{Kostant-73,Atiyah1982,GuilleminSternberg1982,GuilleminSternberg1984,Brion-06,ressayre2011geometric,vergne2017inequalities,walter2014multipartite,van2025computing}).
The quantum functionals~$F_\theta$ satisfy Strassen's four axioms for the family of all $d$-tensors and hence are universal spectral points, generalizing the flattening ranks, which are recovered when $\theta$ is a deterministic distribution ($\theta = e_i$ for some $i\in[d]$).
The similarity between the definitions of~$\zeta^\theta$ and~$F_\theta$ is remarkable and it is natural to ask if these quantities are related.
In \cite{CVZ2018}, the authors proved that~$F_{\theta}(t)\leq \zeta^{\theta}(t)$, using a nontrivial relationship between the support polytope and moment polytope of a tensor due to Strassen~\cite{Strassen-05}, and they established equality for the (non-generic) class of \emph{free} tensors that includes the oblique tensors.
However, the precise relationship between the two functionals was not known, and indeed, the question of whether Strassen's support functionals are universal spectral points has remained open for over 35 years.


In this paper we answer this question in the affirmative by proving that the quantum functionals~$F_\theta(t)$ and Strassen's upper support functional~$\zeta^\theta(t)$ are in fact equal, for \emph{every} tensor $t$ and every~$\theta\in\Theta$.

\begin{thm}\label{thm:support-equals-quantum}
For every tensor $t$ and every $\theta\in\Theta$, $F_\theta(t)=\zeta^\theta(t)$.
In particular, Strassen's support functional is a universal spectral point in the asymptotic spectrum of tensors.
\end{thm}

Our result also provides a new and direct proof that the quantum functionals~$F_\theta$ are universal spectral points.
We expand more on this in \cref{sec:intro-applications}.

\Cref{thm:support-equals-quantum} also derives a new equality between the \emph{asymptotic slice rank} of tensors and the asymptotic vertex cover number of hypergraphs.
The \emph{slice rank} $\SR(t)$ of a tensor~$t$ is the minimum number of slices that sum to~$t$, where a tensor is called a \emph{slice} if it has a principal flattening of rank one.
It was introduced by Tao and has various applications in algebraic complexity and additive combinatorics, most notably for bounding the sizes of \emph{cap sets}~\cite{Tao-16,Tao-Sawin-16}.
As the slice rank is not multiplicative, one defines the \emph{asymptotic slice rank} as its regularization (which always exists for complex tensors~\cite[Cor.~5.7]{CVZ2018}):
\[
    \widetilde{\SR}(t) \coloneqq \lim_{n\rightarrow\infty} \SR(t^{\otimes n})^{1/n}.
\]
To any $d$-tensor $t$, we assign also a $d$-uniform, $d$-partite hypergraph~$\HH_t$ with hyperedges~$(i_1,i_2,\dots,i_d)\in E(\HH_t)$ if and only if $t_{i_1,i_2,\dots,i_d}\neq 0$.
We denote by $\tau(\HH)$ the \emph{vertex cover (or hitting set) number} of a hypergraph~$\HH$, and define the \emph{asymptotic vertex cover number} by
\begin{align*}
    \tilde{\tau}(\HH) \coloneqq \lim_{n\rightarrow\infty} \tau(\HH^{\times n})^{1/n},
\end{align*}
where $\HH^{\times n}$ denotes the $n$-th power of the hypergraph (see \cref{sec:applications} for the definitions).
We find:

\begin{cor}\label{cor:slice-rank}
    For every tensor~$t$, the asymptotic slice rank can be computed as
    $\displaystyle\widetilde{\SR}(t) = \min_{g\in\GL} \tilde{\tau}(\HH_{g\cdot t})$.
\end{cor}

\noindent
We also give a generalization of this result to the weighted asymptotic slice rank, see \cref{thm:intro-weighted-sr}.
Previously, such a formula was only known for free tensors~\cite{CVZ2018,CLZ-23}.
\Cref{cor:slice-rank} and its generalization follows by a direct application of \cref{thm:support-equals-quantum}, because the asymptotic slice rank and the asymptotic vertex cover number can be computed by minimization of the quantum and support functional, respectively, over the parameter~$\theta$:
we have $\widetilde{\SR}(t) = \min_{\theta\in\Theta} F_\theta(t)$ and~$\tilde{\tau}(\HH) = \min_{\theta\in\Theta} \zeta^\theta(t)$~\cite{CVZ2018,Tao-Sawin-16}.

\subsection{Classical-quantum correspondences for convex optimization on moment polytopes and beyond}
\label{intro:minimax}
There is a deep connection between convex optimization on entanglement (and more general moment) polytopes and several central problems in theoretical computer science, quantum physics, and invariant theory, which has been uncovered in a series of works by many authors over the past decade.
To highlight a few examples:
scaling algorithms such as \emph{matrix}, \emph{operator}, and \emph{tensor scaling} have been understood to link norm minimization over group orbits and moment polytopes of an input matrix, quantum channel, which is the starting point of the emergent theory of non-commutative optimization~\cite{AZGLOW,BFGOWW_FOCS2019,HiraiSakabe2024,Hirai_ncrank}.
In invariant theory, by building on theorems by Hilbert--Mumford and Kempf--Ness, the same computational problem captures central notions of \emph{stability}.
In quantum physics, the \emph{single-particle quantum marginal} problem in quantum physics---namely, determining when a collection of given reduced density matrices is realizable as the marginals of a global state, and to optimize over such marginals---is captured precisely by membership testing and optimization over entanglement polytopes~\cite{kly1,kly2,walter2013entanglement,walter2014multipartite}.
In algebraic complexity, a broad range of tensor parameters admit characterizations as convex optimization problems over the entanglement polytope.
We already saw this for the quantum functionals, which are defined as entropy maximization problem over the entanglement polytope; other examples include the \emph{$G$-stable rank} and the \emph{non-commutative rank} which we discuss in \cref{sec:intro-applications}.

Despite its significance, efficient optimization over the moment polytopes has proved to be a major challenge.
An indication of this challenge, moment polytopes~$\Delta(t)$ are specified only implicitly by the input tensor~$t$, they may (and will) in general have a number of vertices and facets that is exponential in the input size---at the same time, this strongly motivates the development of new algorithmic approaches for this important class of problems.
At present, no algorithms are known that optimize a convex objective over the moment polytopes in time polynomial in the input size.
This calls for an improved structural understanding and alternative approaches of convex optimization over moment polytopes.

Our results can be viewed as a step towards overcoming this difficulty (see also \cref{intro:outlook}).
Indeed, \cref{thm:support-equals-quantum} shows that the problem of optimizing entropy over the entanglement polytope can be rephrased in terms of a family of convex optimization problems over the simple-to-compute support polytopes (each of which can in principle be carried out in polynomial time using standard methods).
In fact, \cref{thm:support-equals-quantum} is a special case of a result that applies to \emph{any} convex, symmetric and lower semi-continuous \emph{(l.s.c.)} function~$F \colon \RR^{n_1}\times\dots\times\RR^{n_d} \to \RR \cup \{\infty\}$ on the entanglement polytope of a tensor, or, in fact, on \emph{moment polytopes} associated with arbitrary actions of the group~$\GL$.
Here, by \emph{symmetric} we mean that the function~$F$ is invariant under permuting the entries within each argument, that is, under the permutation group~$S_{n_1}\times\dots\times S_{n_d}$.
Our theorem states that for any such function, one obtains the same value by either minimizing~$F$ over the moment polytope~$\Delta(t)$, or by minimizing~$F$ over the support polytopes~$\Omega(g \cdot t)$ for every~$g \in \GL$ and taking the~maximum~value:

\begin{thm}[Minimax formula for moment polytope optimization]
\label{thm:moment-polytope-general-q}
Let $F \colon \RR^{n_1} \times \dots \times \RR^{n_d} \to \RR \cup \{\infty\}$ be a convex symmetric l.s.c.\ function with $\Delta(t)\subseteq\dom F$.
Then, for any tensor~$0 \neq t \in \CC^{n_1} \ot \cdots \ot \CC^{n_d}$,
\begin{equation}
\label{eq:moment-support-duality}
    \min_{p \in \Delta(t)} F(p)
=   \max_{g \in \GL} \min_{p \in \Omega(g\cdot t)} F(p).
\end{equation}
More generally, \eqref{eq:moment-support-duality} holds for any rational representation~$\pi \colon \GL\rightarrow\GL(\VV)$ of the group~$\GL \coloneqq\GL(n_1)\times\cdots\times\GL(n_d)$ and vector~$0\neq t\in \VV$,
where $\Delta(t)$ denotes the moment polytope and $\Omega(g \cdot t)$ the support polytope of~$\pi(g)t$.
\end{thm}

\noindent
We prove \cref{thm:moment-polytope-general-q} as \cref{thm:generalpolytope} in \cref{sec:moment-poly}, where we give all relevant definitions and also show the maximization can be restricted to the unitary subgroup $\U \coloneqq \U(n_1)\times\dots\times \U(n_d)$ (\cref{rem:K general mopo}).
\Cref{thm:support-equals-quantum} is a direct consequence of \cref{thm:moment-polytope-general-q} for the choice~$F(p)\coloneqq -\sum_{i}\theta_iH(p_i)$.
In \cref{sec:intro-applications} we show that \cref{thm:moment-polytope-general-q} captures and relates several further tensor invariants, including the non-commutative rank~\cite{Fortin-Rautenauer-04, GGOW16, AZGLOW, IQS, IQS2, DM-oc, DM-poly, Hirai_ncrank}, the $G$-stable rank~\cite{Derksen-22, Hirai2025}, and the symmetric quantum functional~\cite{CFTZ-22}, leading to new connections and~interpretations.

The proof of \cref{thm:moment-polytope-general-q} relies on a powerful duality theorem for convex optimization on Hadamard manifolds due to Hirai~\cite{Hirai2025}.
A \emph{Hadamard manifold} is a complete, simply-connected Riemannian manifold with non-positive curvature.
Just like in Euclidean space, Hadamard manifolds have the convenient property that any two points are connected by a unique geodesic (shortest path).
On such manifolds $\MM$, we consider \emph{geodesically convex} functions $f \colon \MM\rightarrow\RR$ that are convex along every geodesic in~$\MM$.
Hirai's recent result establishes a strong Fenchel-type duality in this setting.
Specifically, assume we have a twice differentiable geodesically convex function $f\colon\MM\rightarrow\RR$, and a function $Q\colon T^* \MM\rightarrow\RR$, defined on the cotangent bundle of $\MM$, such that for every base point~$x\in \MM$, $Q_x\coloneqq Q|_{T^*_x\MM}$ is a convex l.s.c.\ function with bounded sublevel sets.
In addition, $Q$ should be invariant under \emph{parallel transport}, meaning that~$Q_y$ can be obtained from~$Q_x$ by transporting cotangent vectors along curves from~$x$ to~$y$.
In particular~$Q$ is already fully determined by~$Q_x$ for any fixed base point~$x\in\MM$.
Under these assumptions,
Hirai constructs a dual maximization program to the minimization problem:
\[
\inf_{x\in \MM} Q(\diff f(x))
\] in terms of the Legendre--Fenchel conjugate of $Q$ and the \emph{recession function} of $f$.
Hirai~\cite[Thm.~3.19]{Hirai2025} proved that the weak duality always holds (the infimum is at least the supremum), and proved the strong duality (the infimum equals the supremum of the dual) also holds for a large class of Hadamard manifolds including the manifolds of positive definite matrices and their products, under a mild continuity assumption on $Q$.
In our context, we apply Hirai's theorem to~$\MM=\PD\coloneqq\PD(n_1)\times\dots\times\PD(n_d)$, where~$\PD(n)$ denotes the manifold of $n\times n$ positive definite Hermitian matrices (equipped with the affine-invariant Riemannian metric).
Since~$\PD$ is an open subset of $\Herm(n_1) \times \dots \times \Herm(n_d)$, the cotangent space of $\PD$ at the identity element is naturally identified as $T_I^* \PD\cong \Herm(n_1) \times \dots \times \Herm(n_d)$.
Then parallel transport invariance property of~$Q$ translates to the assumption that $Q_I \colon \Herm(n_1) \times \dots \times \Herm(n_d) \rightarrow\RR$ is $\U$-invariant, where $\U \coloneqq \U(n_1)\times\dots\times \U(n_d)$, that is, it only depends on the eigenvalues of the Hermitian matrices.
In other words, there exists a function $F \colon \RR^{n_1} \times \cdots \times \RR^{n_d} \rightarrow\RR\cup\{\infty\}$ such that
\[
    Q_I (H_1,\dots,H_d) \equiv F(\spec(H_1),\dots,\spec(H_d)).
\]
Hirai's assumption on~$Q$ translates to the assumption that~$F$ is a convex symmetric l.s.c.\ function with bounded sublevel sets, and conversely any such~$F$ gives rise to a function~$Q$ satisfying Hirai's assumptions, via the above formula.
In this setting, we use Hirai's strong duality theorem to prove the following result, which implies \cref{thm:moment-polytope-general-q} and generalizes it beyond moment polytopes:

\begin{thm}[Minimax formula for gradient optimization]
\label{thm:general-theorem-intro}
Let~$F \colon \RR^{n_1} \times \cdots \times \RR^{n_d} \to \RR\cup\{\infty\}$ be a convex symmetric l.s.c.\ function with bounded sublevel sets,
and let $f \colon \PD \to \RR$ be a twice-differentiable geodesically convex function such that~$\diff f(\PD) \subseteq \dom Q$, with $Q \colon T^* \PD\rightarrow\RR\cup\{\infty\}$ as above.
Then, under the mild continuity assumption~\ref{it:C} described in \cref{sec:general},
\[
\inf_{x\in \PD} Q(\diff f(x)) = \sup_{g\in \GL} \inf_{y\in\RR^{n_1} \times \dots \times \RR^{n_d}}  F(\nabla f^g(y)),
\]
where~$f^g \colon \RR^{n_1} \times \cdots \times \RR^{n_d}\rightarrow\RR$ for $g\in\GL$ denotes the function~$f^g(y) \coloneqq f(g_1^\dag e^{\diagmat y_1}g_1, \dots, g_d^\dag e^{\diagmat y_d}g_d)$.
\end{thm}

\noindent
We prove \cref{thm:general-theorem-intro} as \cref{thm:generaltheorem} in \cref{sec:general}, where we also note that the supremum over~$\GL$ can be reduced to~$\U$ (\cref{rem:K_general}).
\Cref{thm:moment-polytope-general-q} is obtained from \cref{thm:general-theorem-intro} by specializing~$f$ to the \emph{Kempf--Ness function} of~$t$; this is the geodesically convex function $f_t \colon \PD\rightarrow\RR$ defined by $f_t(x) = \log \langle t, x \cdot t \rangle$, where $\langle\cdot,\cdot\rangle$ is the~$\ell^2$-inner product.
The derivative of the Kempf--Ness function at the identity is known as the \emph{moment map}; its image determines the moment polytope~$\Delta(t)$, see~\cref{sec:moment-poly}.

\subsection{Applications to tensor invariants}
\label{sec:intro-applications}

The minimax formula in \cref{thm:moment-polytope-general-q} extend beyond the equality of the quantum and the support functionals.
In this section, we show its versatility by deriving new connections between several tensor invariants.
Specifically, \cref{thm:moment-polytope-general-q} yields
\begin{itemize}
    \item a description of the asymptotic weighted slice rank of a tensor in terms of the asymptotic weighted vertex cover number of its associated hypergraph,
    \item a simpler proof that the quantum functionals are universal spectral points,
    \item a new characterization of the symmetric quantum functional in terms of support polytopes, which leads to a simpler proof of its submultiplicativity,
    \item and a unified derivation of the moment polytope characterizations of Derksen's $G$-stable rank and the non-commutative rank.
\end{itemize}
Each of these applications essentially come from the same idea: By choosing an appropriate l.s.c.\ convex, symmetric function $F$, the left-hand side of the minimax formula of \cref{thm:moment-polytope-general-q} captures the moment polytope characterization of a  tensor invariant, whereas the right-hand side (optimization over the support polytopes) reveals the underlying combinatorial quantity.

\paragraph*{Weighted slice rank}
The notion of slice rank of a tensor was introduced by Tao (and Sawin) \cite{Tao-Sawin-16} in the context of the cap-set problem in combinatorics.
Concretely, a tensor~$t\in\CC^{n_1}\otimes\dots\otimes\CC^{n_d}$ has slice rank $\SR(t)$ equal to the minimum number of slice-rank $1$ tensors (tensors who has a rank-one principal flattening) that sum to $t$.
A natural generalization, introduced in \cite{CLZ-23}, is the \emph{weighted slice rank}.
For a weight vector $\xi\in\Xi\coloneqq \{\xi\in\RR_{\geq 0}^d\mid \max_i\xi_i=1\}$, the $\xi$-weighted slice rank is defined as
\begin{align*}
\SR_{\xi}(t) \coloneqq &\min\;  r_1^{1/\xi_1}+\dots+r_d^{1/\xi_d}
\end{align*}
where the minimum is over all tuples $(r_1,\dots,r_d)$ such that $t$ has a slice decomposition of size $(r_1,\dots,r_d)$.
This notion interpolates between several tensor invariants:
For instance, choosing $\xi=(1,\dots,1)$ we recover the slice rank, while other choices recover flattening ranks and the non-commutative rank.

As the ($\xi$-weighted) slice rank is neither sub- nor super-multiplicative, it is only natural to consider its asymptotic regularization \[
\widetilde{\SR_{\xi}}(t) \coloneqq \lim_{n\rightarrow\infty}\, \SR_{\xi}(t^{\otimes n})^{1/n},
\]
which is known to exist, see \cite{CVZ2018,CLZ-23}.
A key result is that the asymptotic weighted slice rank admits a moment polytope characterization as an optimization problem over the moment polytope: \[
\widetilde{\SR_{\xi}}(t) = \max_{p\in\Delta(t)} \min_i 2^{H(p_i)/\xi_i},
\] where $H$ denotes the Shannon entropy.

On the other hand, we can associate a $d$-uniform, $d$-partite hypergraph $\HH_t$ to $t$, whose edges correspond to the non-zero entries of $t$.
The $\xi$-weighted vertex cover number $\tau_\xi(\HH_t)$ is defined by a natural covering integer program (\cref{def:vertex-cover}) and its asymptotic regularization is \[
\tilde{\tau}_{\xi}(\HH_t) \coloneqq \lim_{n\rightarrow\infty} \tau_{\xi}(\HH_t^{\times n})^{1/n}
\] (where $\HH_t^{\times n}$ denotes the $n$-fold Kronecker product).
As a natural generalization of Tao and Sawin's proof for the unweighted case, we show in \cref{prop:tao-sawin-weighted} that it satisfies
 \[
\tilde{\tau}_{\xi}(\HH_t) = \max_{p\in\Omega(t)}\, \min_{i} \, 2^{H(p_i)/\xi_i}.
\] Hence, the asymptotic weighted slice rank and vertex cover numbers are obtained by the optimization of the same function, but over different polytopes: the moment polytope $\Delta(t)$ and the support polytope $\Omega(t)$, respectively.

Then, by an application of \cref{thm:moment-polytope-general-q} we obtain:
\begin{thm}
\label{thm:intro-weighted-sr}
   The asymptotic weighted slice rank can be computed as $\widetilde{\SR_{\xi}}(t) = \min_{g\in\GL} \tilde{\tau}_{\xi}(\HH_{g\cdot t})$.
\end{thm}
\noindent
\cref{cor:slice-rank} is recovered as a special case $\xi=(1,1,\dots,1)$.

\paragraph*{Quantum functionals as spectral points}
\cref{thm:support-equals-quantum} shows that Strassen's upper support functional $\zeta^\theta$ and the quantum functional $F_\theta$ coincide for every probability vector~$\theta$.
An immediate consequence that we discuss here is that $F_\theta$ inherits the axioms of a universal spectral point:
Indeed, it is easy to see using basic properties of the entropy and the definition of the polytopes~$\Delta(t)$ and~$\Omega(t)$ that~$F_\theta$ is super-additive and super-multiplicative (as also shown in~\cite{CVZ2018}) and that~$\zeta^\theta$ is sub-additive and sub-multiplicative (as already shown in~\cite{Strassen-91}).
Thus, \cref{thm:support-equals-quantum} implies at once that the quantum functionals are additive and multiplicative.
Monotonicity and normalization are immediate from the definition, concluding the proof.
We reproduce this result in detail in \cref{sub:quantum functionals}.

In contrast, the original proof in \cite{CVZ2018} proceeded in two steps:
after establishing that $F_\theta$ is super-additive and super-multiplicative, they defined an \emph{upper quantum functional}~$F^\theta$ in terms of invariant theory, used representation-theoretic techniques to prove that~$F^\theta$ is sub-additive and sub-multiplicative, and proved that $F_\theta=F^\theta$ by comparing the analytic vs.\ invariant theoretic characterizations of the entanglement polytope.
\Cref{thm:support-equals-quantum} allows us to completely bypass this second step, replacing invariant-theoretic machinery with more elementary convex-analytic methods.
We remark that several geometric proofs that avoid representation theory have been found recently; Hirai~\cite{Hirai2025} uses the Fenchel-type duality to prove sub-multiplicativity, but uses a continuous-time flow to prove sub-additivity; another proof in \cite{DSW-26} is based on the behavior of a discretized variant of his flow.

\paragraph*{Symmetric quantum functional}
In \cite{CFTZ-22}, the authors defined a variant of the quantum functionals for the diagonal action of $\GL(n)$ on $(\CC^n)^{\otimes d}$ as an entropy optimization problem over the moment polytope $\Delta_{\Sym}(t)$ for the diagonal action:
\[
F_{\Sym}(t) \coloneqq \max_{p\in\Delta_{\Sym}(t)} \, 2^{H(p/d)}.
\] It was proved that $F_{\Sym}$ is a spectral point in the asymptotic spectrum of symmetric tensors, and it is sub-multiplicative on all tensors, i.e., $F_{\Sym}(t\otimes s)\leq F_{\Sym}(t) F_{\Sym}(s)$ for non-zero tensors~$t$ and~$s$.

An application of \cref{thm:moment-polytope-general-q} to the symmetric, concave function $F(p) \coloneqq H(p/d)$ yields a new description of the symmetric quantum functional in terms of the support polytopes:
\begin{thm}
\label{cor:symmetric-functional}
    For every tensor $0\neq t\in (\CC^n)^{\otimes d}$, we have \[
    F_{\Sym}(t) = \min_{g\in\GL(n)}\, \max_{p\in \Omega_{\Sym}(g\cdot t)} \, 2^{H(p/d)},
    \] where $\Omega_{\Sym}$ is the support polytope with respect to the diagonal action.
\end{thm}

An interesting application of this new characterization of $F_{\Sym}$ is a new and simpler proof of its sub-multiplicativity.
From a simple description of the support polytope $\Omega(t\otimes s)$ in terms of~$\Omega(t)$ and~$\Omega(s)$, we prove in \cref{prop:symmetric-sub-mult} that $F_{\Sym}$ is sub-multiplicative.
In contrast, the original proof in \cite{CFTZ-22} relies on a highly non-trivial invariant-theoretic relationship between the moment polytopes $\Delta_{\Sym}(t\otimes s)$ and $\Delta_{\Sym}(t),\Delta_{\Sym}(s)$.

\paragraph*{$G$-stable rank and non-commutative rank}
Derksen's $G$-stable rank \cite{Derksen-22} and the non-commutative rank of tensors are two notions of rank that found various applications in algebraic complexity theory.
Derksen, for the $G$-stable rank, and Hirai \cite{Hirai_ncrank}, for the non-commutative rank, proved formulas for them as optimization problems over the moment polytope.
The minimax formula (\cref{thm:moment-polytope-general-q}) yields a general and streamlined strategy to reproduce these results, which may be used to prove moment polytope characterizations for other tensors invariants.

Derksen's $G$-stable rank \cite{Derksen-22} of tensors is defined as an optimization problem over the one-parameter subgroups of~$\GL$ and admits an LP formulation as the minimum fractional vertex cover number of the associated hypergraph $\HH_{g\cdot t}$ over all possible basis choices $g\in\GL$.
For a weight vector $\alpha\in\RR^d_{>0}$, the \emph{$\alpha$-weighted $G$-stable rank} $\rk^G_{\alpha}(t)$ of a tensor $t$ can be defined as \[
\rk^G_\alpha(t) \coloneqq \min_{g\in\GL} \, \tau_{\alpha}^f(\HH_{g\cdot t})
\] where $\tau_{\alpha}^f(\HH)$ denotes the $\alpha$-weighted \emph{fractional vertex cover} of $\HH$, see \cref{def:G-rk-fractional}.
Using LP duality, one can express the fractional vertex cover number of $\HH_t$ as an optimization problem over the support polytope $\Omega(t)$ (\cref{prop:fractional-cover-reciprocal}):
\[
\tau_\alpha^f(\HH_t) = \max_{p\in\Omega(t)} \, \min_{i\in [d]}\, \frac{\alpha_i}{\norm{p_i}_{\infty}}.
\] The function $F(p)\coloneqq \min_{i\in [d]}\frac{\alpha_i}{\norm{p_i}_{\infty}}$ is symmetric and concave.
Applying the minimax \cref{thm:moment-polytope-general-q} gives
\[
\rk^G_{\alpha}(t) = \max_{p\in\Delta(t)} \min_{i\in [d]} \frac{\alpha_i}{\norm{p_i}_{\infty}},
\] which is precisely Derksen's moment polytope characterization of the $G$-stable rank (\cite[Theorem~5.2]{Derksen-22}):

Likewise, we can reproduce Hirai's moment polytope characterization \cite{Hirai_ncrank} of the non-commutative rank of a matrix tuple.
For a $3$-tensor $t\in\CC^{n}\otimes\CC^n\otimes\CC^m$, viewed as a matrix tuple $t=(A_1,A_2,\dots,A_m)$,
the non-commutative rank $\ncrk(t)$ is the rank of the matrix $x_1 A_1+\dots+x_m A_m$ over the free skew field $\CC\langle x_1,\dots,x_m\rangle$.
The Fortin-Reutenauer formula \cite{Fortin-Rautenauer-04} expresses it as \[
\ncrk(t) = \min_{(g,h)\in\GL(n)\times\GL(n)} \; \tau(\BB_{(g\otimes h\otimes I_m) t})
\]
where $\tau(\cdot)$ denotes the ordinary vertex cover number of a graph and $\BB_{(g\otimes h\otimes I_m)t}$ is the bipartite graph associated with the matrix tuple $(g\otimes h\otimes I_m) t$, whose vertex set is $[n]\sqcup [n]$ and has an edge $(i,j)\in [n]\times [n]$ whenever there exists $k$ with~$(A_k)_{ij}\neq 0$.

By K\H{o}nig's theorem and LP duality, we show in \cref{prop:l1-vertex-cover} that the vertex cover number admits a formula as an optimization problem over the support polytope of $t$:
\[
\tau(\BB_t) = n - \frac{n}{2} \min_{p\in\Omega_{\mathrm{LR}}(t)} \, \norm{p - \frac{\mathbf{1}_{2n}}{n}}_1,
\] where $\Omega_{\mathrm{LR}}(t)$ denotes the support polytope of $t$ associated to the left-right action of $\GL(n)\times \GL(n)$ on $\CC^{n}\otimes\CC^n\otimes\CC^m$ and $\mathbf{1}_n\in\RR^n$ denotes the all-ones vector.
Note that $F(p)\coloneqq \norm{p-\frac{\mathbf{1}_{2n}}{n}}_1$ is a symmetric convex function of $p$.
Applying the minimax \cref{thm:moment-polytope-general-q} to $F$ yields Hirai's moment polytope characterization of the non-commutative rank \cite{Hirai_ncrank}: \[
\ncrk(t) = n-\frac{n}{2}\, \min_{p\in\Delta_{\mathrm{LR}}(t)}\, \norm{p-\frac{\mathbf{1}_{2n}}{n}}_1.
\]
\subsection{Outlook}\label{intro:outlook}
We indicate some directions for future research.
While \cref{thm:moment-polytope-general-q,thm:general-theorem-intro} are stated for actions of $\GL = \GL(n_1) \times \dots \times \GL(n_d)$ and geodesically convex function on $\PD = \PD(n_1)\times\dots\times\PD(n_d)$, their proofs are general.
Hirai conjectured that strong duality holds for \emph{every} Hadamard manifold~$\MM$ \cite[Conj.~3.18]{Hirai2025}.
If this conjecture is true, \cref{thm:moment-polytope-general-q} generalizes to arbitrary rational representations $\pi \colon G\rightarrow\GL(\VV)$ of self-adjoint subgroups~$G\leq\GL(n)$.
Likewise, \cref{thm:general-theorem-intro} generalizes to the symmetric space~$\MM\coloneqq\{g^\dag g\mid g\in G\}$, which is a totally geodesic submanifold of~$\PD(n)$, and indeed admits a natural generalization to any Hadamard manifold.

It is interesting that our results lead to a short and direct proof that the quantum functionals are universal spectral point, as shown in \cref{sec:intro-applications}.
We speculate that the general minimax formula in \cref{thm:general-theorem-intro} might be a useful ingredient for the construction of spectral points.

Finally, we comment on the problem of \emph{computing} the quantum functionals, or, more generally, computing $\min_{p\in\Delta(t)}F
(p)$ for a convex function~$F$ by efficient algorithms.
If one has access to the moment polytope (either as a convex hull of vertices, an intersection of half-spaces or by a membership or separation oracle), then the minimum can be computed by standard methods of convex optimization.
No such simple description is known that is efficient in all parameters, but see \cite{GGOW16,burgisser2018efficient,BFGOWW_FOCS2019,hirai2023interior,van2025computing} for progress.
Hirai showed that there is a natural \emph{$Q$-gradient}~$\nabla^Q f$ in the setting described above and proved that the corresponding flow~$\dot{x} = - \nabla^Q f(x)$ minimizes~$Q(\diff f(x))$~\cite{Hirai2025}, generalizing an earlier analysis of the ordinary gradient flow~\cite{HiraiSakabe2024}.
In~\cite{DSW-26}, we show that Hirai's flow can be discretized and analyzed for a natural class of~$Q$ and~$f$.
In particular, we find the quantum functionals~$F_\theta(t)$ can be computed by a simple and natural \emph{entropic scaling algorithm}: $t \leftarrow ( \rho_1^{-\theta_1/2} \otimes \rho_2^{-\theta_2/2} \otimes \dots\otimes \rho_d^{-\theta_d/2}) t$, where~$\rho_i \coloneqq {t_i^\dag t_i}/{\norm{t}^2}$ for~$i\in [d]$, which converges in a number of steps that is polynomial in $n_1,\dots,n_d$ and the bit-length of~$t$.
It would be interesting to use the minimax formula in \cref{thm:moment-polytope-general-q} to obtain an alternative approach for this problem.

\subsection{Organization of the paper}
The preliminary \cref{sec:prelim} includes the relevant definitions, such as Hadamard manifolds and the Legendre-Fenchel conjugate of a convex function.
\cref{sec:general} is devoted to the Hirai's asymptotic duality over Hadamard manifolds and the proof of the general minimax formula (\cref{thm:general-theorem-intro}).
In \cref{sec:moment-poly}, we use \cref{thm:general-theorem-intro} to prove \cref{thm:moment-polytope-general-q}.
The proof is built upon the characterization of the moment polytope as the set of gradients of the Kempf-Ness function.
\cref{sec:applications} gives various applications of \cref{thm:general-theorem-intro,thm:moment-polytope-general-q} to several tensor invariants, such as asymptotic weighted slice rank, $G$-stable rank, and non-commutative rank.
In this section we also expand on the new and simpler proof of quantum functionals as spectral points.

\section{Preliminaries}
\label{sec:prelim}
In this paper, we write $[n] \coloneqq \{1, 2, \ldots, n\}$.
For a subset~$S \subseteq \RR^n$, let $\overline{S}$ denote its closure and $\conv S$ denote its convex hull.
We write $\RR^n_{\downarrow} \coloneqq \{(y_1, \ldots, y_n) \in \RR^n \mid y_1 \ge \cdots \ge y_n\}$ for the closed convex set of vectors that are sorted in non-increasing order.

\subsection{Linear algebra}
All vector spaces are finite-dimensional by assumption.
For a $\CC$-vector space $\VV$, $\Lin(\VV)$ denotes the set of linear operators on~$\VV$, and $\GL(\VV) \subseteq \Lin(\VV)$ denotes the group of invertible ones.
If $\VV = \CC^n$, we simply write~$\Lin(n)$ and~$\GL(n)$, which are identified by the set of $n \times n$ (nonsingular) matrices.
We denote by~$I \in \GL(n)$ the identity matrix.
Let $\Herm(n) \subseteq \Lin(n)$ and $\U(n) \subseteq \GL(n)$ denote the set of Hermitian and unitary matrices, respectively.
For a vector $y \in \CC^n$, $\diagmat y \in \Lin(n)$ denotes the diagonal matrix whose $i$-th diagonal entry is~$y_i$, while if~$A \in \Lin(n)$ is a matrix then $\diagvec A \in \CC^n$ denotes the vector consisting of its diagonal entries.
For~$A \in \Lin(n)$, let~$e^{A} \in \GL(n)$ denote the matrix exponential.
Let~$\spec X \in \RR^n_{\downarrow}$ denote the spectrum of $X \in \Herm(n)$, i.e., the vector of sorted eigenvalues of $X$.
Let~$X \oplus Y \in \Lin(m + n)$ denote the block-diagonal matrix with diagonal blocks~$X \in \Lin(m)$ and~$Y \in \Lin(n)$.

\subsection{Hadamard manifolds}\label{sub:hadamard_manifolds}
For a Riemannian manifold $\MM$, let $T_x\MM$ denote the tangent space at $x \in \MM$, $T^*_x\MM$ denote the cotangent space, and~$T\MM \coloneqq \sqcup_{x \in \MM}T_x\MM$ and $T^*\MM \coloneqq \sqcup_{x \in \MM}T_x^*\MM$ denote the (co)tangent bundles.
Let~$\tau_{x \to y}\colon T_x\MM \to T_y\MM$ denote the parallel transport along the (unique) geodesic, and~$\tau^*_{x \to y}\colon T^*_y \to T^*_x$ be its dual map, which parallel transports cotangent vectors.
Let $\Gamma(\MM)$ denote the set of all geodesic rays, i.e., geodesic curves defined on~$[0,\infty)$.

A \emph{Hadamard manifold} is a complete, simply-connected Riemannian manifolds with non-positive curvature.
We do not introduce the general definition of these concepts on manifolds, since we only consider Euclidean spaces and products of positive definite matrices, which we discuss in the following.

\paragraph{Euclidean space}
In the case of the Euclidean space~$\MM = \RR^n$ with its usual structure as a Riemannian manifold, the (co)tangent spaces $T_x\MM$ and $T^*_x\MM$ are identified with~$\RR^n$, and the parallel transport is the identity map.
A geodesic ray $\gamma(t) \in \Gamma(\RR^n)$ is a line $\gamma(t) = q + tp$ ($q, p \in \RR^n$).
For a function $f: \RR^n \to \RR \cup \{\infty\}$, the gradient at $x$ is denoted by~$\nabla f(x) \in \RR^n$, if it exists.

\paragraph{Positive-definite matrices with the affine-invariant Riemannian metric}
Let $\PD(n)$ denote the set of $n \times n$ positive-definite Hermitian matrices.
It is an open subset of~$\Herm(n)$, and hence we can canonically identify the tangent space $T_x \PD(n) \cong \Herm(n)$ at every point $x \in \PD(n)$.
Under this identification, the \emph{affine-invariant Riemannian metric} on $\PD(n)$ is given by~$\langle X,Y\rangle_x := \tr[x^{-1} X x^{-1} Y]$ for $x \in \PD(n)$ and $X, Y \in \Herm(n)$ and turns $\PD(n)$ into a Hadamard manifold~(e.g., \cite[\S{}II.10]{BridsonHaefliger1999}).
Every geodesic ray $\gamma \in \Gamma(\PD(n))$ is of the form $\gamma(t) = g^\dagger e^{\diagmat(q + tp)}g$ for~$g \in \GL(n)$, $q, p \in \RR^n$, and the parallel transport along the geodesic from the identity~$I \in \PD(n)$ to any~$x \in \PD(n)$ is given by~$\tau_{I\to x}(X) = x^{1/2} X x^{1/2}$, where $X \in \Herm(n) \cong T_I \PD(n)$.
We similarly identify the cotangent spaces~$T_x^* \PD(n) \cong \Herm(n)^* \cong \Herm(n)$, where the first isomorphism is canonical from the identification~$T_x \PD(n) \cong \Herm(n)$ and the second isomorphism uses the Hilbert-Schmidt inner product~$\langle X, Y\rangle_{\Herm(n)} \coloneqq \tr[XY]$.
Then the parallel transport of cotangent vectors along the geodesic from~$x$ to~$I$ is given by $\tau^*_{I \to x}(\hat X) = x^{1/2} \hat X x^{1/2}$, where~$\hat X \in \Herm(n) \cong T_x^* \PD(n)$.
For a differentiable function $f : \PD(n) \to \RR$, its differential at $x \in \PD(n)$ is denoted by $\diff f(x) \in T^*_x\PD(n)$.

We also define the direct product $\PD \coloneqq \PD(n_1) \times \cdots \times \PD(n_d)$.
By embedding $\rmH \coloneqq \Herm(n_1) \op \cdots \op \Herm(n_d)$ as block-diagonal matrices in~$\Herm(N)$, $N \coloneqq n_1 + \cdots + n_d$, we consider $\PD$ as a totally geodesic submanifold of $\PD(N)$.
Hence it is also a Hadamard manifold, its (co)tangent spaces identify with the subspace $\rmH \subseteq \Herm(N)$, and all of the above formulas for $\PD(n)$ can also be used for $\PD$.

\subsection{Convex functions and Legendre--Fenchel conjugate}
A function $f\colon \RR \to \RR \cup \{\infty\}$ on the real line is called \emph{convex} if $f(tx + (1 - t)y) \le tf(x) + (1 - t)f(y)$ for all~$x, y \in \RR$ and $t\in[0,1]$.
A function $f\colon \MM \to \RR \cup \{\infty\}$ on a Hadamard manifold~$\MM$ is called (geodesically) \emph{convex} if $f \circ \gamma$ is convex for every geodesic $\gamma \colon \RR \to \MM$.
The domain of a convex function is defined as $\dom f \coloneqq \{x \in \MM \mid f(x) <\infty\}$.

When $\MM = E$ is an $\RR$-vector space, the above coincides with the usual notion of convexity.
In this setting, for a \emph{proper} convex function, that is, a convex function~$f\colon E \to \RR \cup \{\infty\}$ having a nonempty domain, the \emph{Legendre--Fenchel conjugate}~$f^*\colon E^* \to \RR \cup \{\infty\}$ is defined as~$f^*(\alpha) \coloneqq \sup_{y \in E}\alpha(y) - f(y)$, where $E^*$ is the dual space of~$E$.
If~$E$ is a Euclidean space with inner product $(\cdot,\cdot)$, then this can also be thought of as a function~$f^* \colon E \to \RR\cup\{\infty\}$ on~$E$, given by
$f^*(x) \coloneqq \sup_{y \in E}(x,y) - f(y)$.
We will always be in this situation: either $E = \RR^n$ with the standard inner product, or~$E = \rmH \subseteq \Herm(N)$ with the restriction of the Hilbert-Schmidt inner product.

\section{A minimax formula for gradient optimization}\label{sec:general}
In this section, we present our general minimax formula (\cref{thm:general-theorem-intro}, restated below as \cref{thm:generaltheorem}).
Its proof relies on a recent Fenchel-type duality theorem for convex optimization on Hadamard manifolds, referred to as ``asymptotic duality'' by Hirai~\cite{Hirai2025}.
Hirai's theorem applies to a broader class of Hadamard manifolds, but we only require it in the cases of~$\PD$ and $\RR^N$ (\cref{sub:hadamard_manifolds}).
In the following, we first discuss Hirai's result concretely in this setting and then state our theorem and its proof.

Throughout this section, let $n_1, \ldots, n_d$ be integers and $N \coloneqq n_1 + \cdots + n_d$.
For a function $F \colon \RR^{n_1} \times \cdots\times \RR^{n_d} \cong \RR^N \to \RR \cup \{\infty\}$, consider the following properties:
\begin{enumerate}[leftmargin=4em,noitemsep,label={(F\arabic*)}]
    \item\label{it:F1} $F$ is convex and lower-semicontinuous (l.s.c.).
    \item\label{it:F2} $F$ is symmetric in each argument, i.e., $F(P_1x_1, \ldots, P_dx_d) = F(x_1, \ldots, x_d)$ for all permutation matrices~$P_1, \ldots, P_d$.
    \item\label{it:F3} $F$ has bounded sublevel sets.
\end{enumerate}
As in \cref{sub:hadamard_manifolds}, we let $\rmH \coloneqq \Herm(n_1) \oplus \cdots \oplus \Herm(n_d) \subseteq \Herm(N)$, where the inclusion is by embedding as block-diagonal matrices.
We can extend $F$ to tuples of Hermitian matrices in~$\rmH$ by taking the spectrum of each matrix, as follows:
\begin{gather}
    F_\rmH \colon \rmH \to \RR \cup \{\infty\},\notag\\
    F_\rmH(X_1, \ldots, X_d) \coloneqq F(\spec X_1, \ldots, \spec X_d).\label{eq:dfn_fH}
\end{gather}
In short, $F_\rmH(X) = F(\spec X)$, where $\spec$ here acts argument-wise.
It is clear that $F_\rmH$ is unitarily invariant, i.e., $F_\rmH(u_1^\dagger X_1u_1, \ldots, u_d^\dagger X_du_d) = F_\rmH(X_1, \ldots, X_d)$ for all~$u_i \in \U(n_i)$.
In short, $F_\rmH(u^\dagger Xu) = F_\rmH(X)$ for any~$u \in \U \coloneqq \U(n_1) \times \cdots \times \U(n_d) \subseteq \U(N)$.
Moreover, the following properties hold:
\begin{prop}[{\cite{Davis1957,Lewis1996}}\footnote{These references prove these properties in the case of~$d = 1$, but their proofs generalize readily to~$d \ge 2$.}]\label{prop:symmetricconvex}
    Let $F\colon \RR^N \to \RR \cup \{\infty\}$ be a function that satisfies~\ref{it:F1} and~\ref{it:F2}.
    Then, $F_\rmH\colon \rmH \to \RR \cup \{\infty\}$ satisfies the following properties:
    \begin{enumerate}
        \item[(a)] $F_\rmH$ is a convex and l.s.c.\ function on $\rmH$.
        \item[(b)] The Legendre--Fenchel conjugate is given by $(F_\rmH)^* = (F^*)_\rmH$ (so we can simply write $F_\rmH^*$).
    \end{enumerate}
\end{prop}

\noindent On the other hand, it is easy to verify that every $\U$-invariant l.s.c.\ convex function $F_\rmH$ having bounded sublevel sets is obtained in this manner by taking~$F \coloneqq F_\rmH(\diagmat(\cdot))$.

Finally, we let $\PD \coloneqq \PD(n_1) \times \cdots \times \PD(n_d)$ be equipped with the affine-invariant metric, as explained in \cref{sub:hadamard_manifolds}.
We can then define a function $Q \colon T^*\PD \to \RR \cup \{\infty\}$ on the cotangent bundle (called a ``norm-like function'' in~\cite{Hirai2025}) by the formula
\begin{align}
    Q(\hat Y) \coloneqq Q_x(\hat Y) \coloneqq F_\rmH\mleft(\tau_{I \to x}^*\hat Y\mright) &= F_\rmH\mleft(x^{1/2}\hat Yx^{1/2}\mright)\notag\\
    (x \in \PD, \hat Y &\in T^*_x\PD \cong \rmH),\label{eq:dfn_QF}
\end{align}
where $Q_x\coloneqq Q|_{T^*_x\PD}$ denotes the restriction of~$Q$ to the cotangent space at~$x$ and we use the identification~$T^*_x\PD \cong \rmH$ from \cref{sub:hadamard_manifolds}.
Define $Q^*\colon T\PD \to \RR \cup \{\infty\}$ by taking the Legendre--Fenchel conjugate on each cotangent space, i.e., $Q^*_x \coloneqq Q^*|_{T^*_x\PD} \coloneqq (Q_x)^*$, and using $(T^*_x\PD)^* \cong T_x\PD$.
Concretely, this is given by
\begin{align*}
    Q^*(X)
&= Q^*_x(X)
= \sup_{\hat Y \in T^*_x\PD} \hat Y(X) - Q(\hat Y)\\
&= \sup_{\hat Y \in \rmH} \tr[X \hat Y] - F_\rmH\mleft(x^{1/2}\hat Yx^{1/2}\mright)\\
&= \sup_{\hat Z \in \rmH} \tr\mleft[x^{-1/2} X x^{-1/2} \hat Z \mright] - F_\rmH(\hat Z)\\
&= F_\rmH^*\mleft(x^{-1/2}Xx^{-1/2}\mright)\ (x \in \PD, X \in T_x\PD \cong \rmH).
\end{align*}
In particular, for any geodesic ray $\gamma(t) \coloneqq g^\dagger e^{\diagmat(q + tp)}g$, the value at $-\dot\gamma(t) \in T_{\gamma(t)}\PD$ is given by
\begin{align}\label{eq:Qstarminus}
    Q^*(-\dot\gamma(t)) &=\notag\\
    F_\rmH^*\Big(&-(g^\dagger e^{\diagmat(q + tp)}g)^{-1/2}g^\dagger\notag\\&e^{\diagmat(q + tp)}\diagmat(p)\,g(g^\dagger e^{\diagmat(q + tp)}g)^{-1/2}\Big)\notag\\
    &= F_\rmH^*(-\diagmat p),
\end{align}
where the last equality follows from the unitary invariance of~$F_\rmH$ combined with the fact that $h (h^\dagger h)^{-1/2}$ is unitary for any~$h \in \GL(N)$.
Since the right-hand side does not depend on~$t$, we can define\footnote{The property that $Q^*(-\dot\gamma(t))$ is independent of $t$ holds on every Hadamard manifold; see \cite[Prop.~3.4]{Hirai2025}.}
\begin{align*}
    Q^*(-\dot\gamma) \coloneqq Q^*(-\dot\gamma(t)).
\end{align*}

Let $f \colon \PD \to \RR$ be a twice-differentiable geodesically convex function satisfying $\diff f(\PD) \subseteq \dom Q$, which is equivalent to $\{\tau^*_{I \to x}\diff f(x)\mid x \in \PD\} \subseteq \dom F_\rmH$.
We suppose the following continuity assumption:\footnote{Assumption~\ref{it:C} was not explicitly stated in \cite{Hirai2025}, but the proof implicitly used it.}
\begin{enumerate}[leftmargin=4em,noitemsep,leftmargin=4em,noitemsep,label={(C)}]
    \item\label{it:C} $F_\rmH$ is continuous on $\overline{\{\tau^*_{I \to x}\diff f(x)\mid x \in \PD\}}$.
\end{enumerate}
If we suppose a simpler condition that $F$ is continuous on $\dom F$, then assumption~\ref{it:C} is implied as follows.
Since~$\spec: \rmH \to \RR^N$ is continuous,\footnote{The continuity of $\spec$ can be implied by Weyl's inequality $\|\spec A - \spec B\|_{\infty} \le \|A - B\|_{\mathrm{op}}$ ($\forall A, B \in \Herm(n)$).} $F_\rmH \coloneqq F \circ \spec$ is continuous on $\dom F_\rmH$.
By lower-semicontinuity, $F_\rmH$ is continuous on $\overline{\dom F_\rmH}$.\footnote{This is shown as follows. Take any $X_i \to X^*$ in $\overline{\dom F_\rmH}$. If $F_\rmH(X^*) = \infty$, the lower-semicontinuity of $F_\rmH$ implies $F_\rmH(X_i) \to \infty$. If $F_\rmH(X^*) <\infty$, it suffices to show that only finitely many $i$ satisfy $F_\rmH(X_i) = \infty$. For any $i$ satisfying $F_\rmH(X_i) = \infty$, by lower-semicontinuity, there exists $Y_i \in \dom F_\rmH$ such that $\|Y_i - X_i\| <1/i$ and $F_\rmH(Y_i) > i$. Define a sequence $(Z_i)$ by~$Z_i \coloneqq X_i$ if $F_\rmH(X_i) < \infty$ and $Z_i \coloneqq Y_i$ otherwise. Then we have~$Z_i \to X^*$, and by the continuity on $\dom F_\rmH$ we have $F_\rmH(Z_i) \to F_\rmH(X^*)$. In particular,~$\lim_{i \to \infty}F_\rmH(Z_i)$ must be finite, which implies that $F_\rmH(X_i) = \infty$ holds for only finitely many $i$.}
Hence, $F_\rmH$ is continuous on~$\overline{\{\tau^*_{I \to x}\diff f(x)\mid x \in \PD\}} \subseteq \overline{\dom F_\rmH}$.

Hirai~\cite{Hirai2025} considers the following gradient minimization problem (in a more general setting),
\begin{equation}\label{eq:infproblem}
    \inf_{x \in \PD} Q(\diff f(x)),
\end{equation}
and shows that its optimal value can be expressed by a dual maximization problem.
To state his result, define the \emph{recession function} $f^\infty(\gamma)$~\cite{Hirai2024,KLM2009} on the space of geodesic rays\footnote{The recession function can be, and is usually, defined as a function on equivalence classes of \emph{asymptotic} geodesic rays~\cite{KLM2009, Hirai2024}; see also \cref{rem:K_general}.} as
\begin{gather*}
    f^\infty \colon \Gamma(\PD) \to \RR \cup \{\infty\},\\
    f^\infty(\gamma) \coloneqq \lim_{t \to \infty}\frac1tf(\gamma(t))
= \lim_{t \to \infty}\frac1t \left( f(\gamma(t)) - f(\gamma(0)) \right),
\end{gather*}
where the limit always exists in $\RR \cup \{\infty\}$ because $\frac1t(f(\gamma(t)) - f(\gamma(0)))$ is non-decreasing, due to the geodesic convexity of~$f$.
Then Hirai's result is the following:

\begin{thm}[{Asymptotic duality~\cite[Thm.~3.19]{Hirai2025}, $\PD$ version}]\label{prop:strongduality}
Let $F \colon \RR^N \to \RR$ be a function satisfying \ref{it:F1}--\ref{it:F3}, and let $f \colon \PD \to \RR$ be a twice differentiable geodesically convex function satisfying $\diff f(\PD) \subseteq \dom Q$, where $Q$ is defined as in \cref{eq:dfn_QF}.
Then, under the continuity assumption~\ref{it:C}, the following equality holds:
    \[
        \inf_{x \in \PD}Q(\diff f(x)) = \sup_{\gamma \in \Gamma(\PD)}-f^\infty(\gamma) - Q^*(-\dot\gamma).
    \]
\end{thm}

In the special case that $\PD = \PD(1)^N$, and using the isometry $(\PD(1))^N \cong \RR^N$ given by the entrywise logarithm, one obtains a Euclidean version of the duality theorem.
This Euclidean version can also be understood as a special case of the Fenchel's duality~\cite[Thm.~31.1]{Rockafellar1970}, and hence the twice differentiability of $f$ and conditions \ref{it:F3} and \ref{it:C} are no longer required; see \cite[Ex.~3.6]{Hirai2025}.
Thus:

\begin{prop}[{Asymptotic duality, Euclidean version}]\label{prop:strongduality-Euc}
    Let~$F\colon \RR^N \to \RR \cup \{\infty\}$ be a convex l.s.c.\ function, and let~$h\colon \RR^N \to \RR$ be a differentiable convex function satisfying~$\nabla h(\RR^N) \subseteq \dom F$.
    For any geodesic ray~$\gamma(t) \coloneqq q + tp \in \Gamma(\RR^N)$, let~$h^\infty(\gamma) \coloneqq \lim_{t \to \infty}\frac1th(q + tp)$ and~$F^*(-\dot\gamma) \coloneqq F^*(-p)$.
    Then, the following equality holds:
    \[
        \inf_{y \in \RR^N}F(\nabla h(y)) = \sup_{\gamma \in \Gamma(\RR^N)}-h^\infty(\gamma) - F^*(-\dot\gamma).
    \]
\end{prop}

We now derive our maximin formula for the minimization problem~\eqref{eq:infproblem}, describing it as a supremum of minimization problems, each of which is defined in the Euclidean space.
To this end, we define, for any $g \in \GL \coloneqq \GL(n_1)\times\cdots\times\GL(n_d) \subseteq \GL(N)$, the function
\begin{gather}
    f^g \colon \RR^{n_1} \times \cdots \times \RR^{n_d}\rightarrow\RR,\notag\\
    f^g(y) \coloneqq f(g_1^\dag e^{\diagmat y_1}g_1, \dots, g_d^\dag e^{\diagmat y_d}g_d).\label{eq:fg}
\end{gather}
If we identify $\RR^{n_1} \times \cdots \times \RR^{n_d} \cong \RR^N$, we can also think of~$f^g$ as a function on~$\RR^N$, defined as
\begin{equation}\label{eq:fg alternative}
    f^g(y) \coloneqq f(g^\dagger e^{\diagmat y}g).
\end{equation}
Note that the geodesic convexity of $f$ implies the convexity of~$f^g$, since the map $y \mapsto g^\dagger e^{\diagmat y}g$ maps each line $q + tp$ in~$\RR^N$ to a geodesic $g^\dagger e^{\diagmat(q + tp)}g$ in~$\PD$.
By direct calculation, the gradient of $f^g$ is given by
\begin{align}
    \nabla f^g(y) &= \diagvec\mleft(g\,\diff f\left(x\right)g^\dagger e^{\diagmat y}\mright) = \diagvec\mleft(u\tau^*_{I \to x}\diff f(x)u^\dagger\mright)\notag\\
    &(x \coloneqq g^\dagger e^{\diagmat y}g, u \coloneqq e^{\diagmat y/2}gx^{-1/2} \in \U),\label{eq:gradient_fg}
\end{align}
where we think of $\diff f(x) \in T^*_x\PD \cong \rmH^* \cong \rmH \subseteq \Herm(N)$.
Note that the symmetry and convexity of~$F$ imply that~$F(\diagvec X) \le F(\spec X)$ for every~$X \in \rmH$.\footnote{The symmetry and the convexity of $F$ imply that $F$ is Schur-convex, meaning that $F(x) \le F(y)$ if $x \preceq y$, where $\prec$ denotes majorization in every argument. Then, the Schur--Horn theorem states that $\diagvec(X) \preceq \spec(X)$.}
For $X \coloneqq u\tau^*_{I \to x}\diff f(x)u^\dagger$, we have $F(\nabla f^g(y)) \le F_\rmH(\tau^*_{I \to x}\diff f(x)) = Q(\diff f(x))$.
Therefore, the domain condition $\diff f(\PD) \subseteq \dom Q$ implies that $\nabla f^g(\RR^N) \subseteq \dom F$.

We are now ready to state our general theorem.
It expresses the minimization problem~\eqref{eq:infproblem} as a supremum of minimization problems, each of which is defined on Euclidean space:

\begin{thm}[Detailed version of \cref{thm:general-theorem-intro}]\label{thm:generaltheorem}
    Let~$F \colon \RR^N \to \RR$ be a function satisfying \ref{it:F1}--\ref{it:F3}, and let~$f \colon \PD \to \RR$ be a twice differentiable geodesically convex function satisfying $\diff f(\PD) \subseteq \dom Q$, where $Q$ is defined as in \cref{eq:dfn_QF}.
    Then, under the continuity assumption~\ref{it:C}, the following equality holds:
    \[
        \inf_{x \in \PD} Q(\diff f(x)) = \sup_{g \in \GL}\inf_{y \in \RR^{n_1} \times \cdots\times \RR^{n_d}} F(\nabla f^g(y)).
    \]
    Concretely,
    \begin{align*}
        &\inf_{x \in \PD} F\left(\spec\left(x^{1/2}\diff f(x)x^{1/2}\right)\right)\\&= \sup_{g \in {\GL}}\inf_{y \in \RR^N} F\left(\diagvec\left(g \, \diff f\left(g^\dagger e^{\diagmat y}g\right)g^\dagger e^{\diagmat y}\right)\right).
    \end{align*}
\end{thm}
\begin{proof}
    Throughout the proof we identify $\RR^{n_1} \times \cdots\times \RR^{n_d} \cong \RR^N$.
    We first observe that
    \begin{align}
        &f^\infty\left(t \mapsto g^\dagger e^{\diagmat (q + tp)}g\right) = \lim_{t \to \infty}\frac1tf(g^\dagger e^{\diagmat (q + tp)}g)\notag\\
        &= \lim_{t \to \infty}\frac1tf^g(q + tp) = (f^g)^\infty(t \mapsto q + tp).\label{eq:torusrecession}
    \end{align}
    Then,
    \begin{align*}
        &\qquad\inf_{x \in \PD} Q(\diff f(x))\\
        &= \sup_{\gamma \in \Gamma(\PD)}-f^\infty(\gamma) - Q^*(-\dot\gamma)\\
        &= \sup_{g \in \GL}\sup_{p, q \in \RR^N}-f^\infty\left(t \mapsto g^\dagger e^{\diagmat (q + tp)}g\right) - F_\rmH^*(-\diagmat p)\\
        &= \sup_{g \in \GL}\sup_{p, q \in \RR^N}-(f^g)^\infty(t \mapsto q + tp) - F^*(-p)\\
        &= \sup_{g \in \GL}\inf_{y \in \RR^N} F(\nabla f^g(y)),
    \end{align*}
    where the first equality follows from Hirai's asymptotic duality (\cref{prop:strongduality}), the second follows from the concrete expression of geodesic rays combined with \cref{eq:Qstarminus}, the third follows from \cref{prop:symmetricconvex}~(b) and \cref{eq:torusrecession}, and the last follows from the asymptotic duality in $\RR^N$ (\cref{prop:strongduality-Euc}).
    Finally, the concrete formula is obtained by combining \cref{eq:dfn_fH,eq:dfn_QF,eq:gradient_fg}.
\end{proof}

\begin{rem}\label{rem:K_general}
    We can obtain another expression for problem~\eqref{eq:infproblem} if we use additional facts.
    Two geodesic rays $\gamma, \gamma' \in \Gamma(M)$ are said to be \emph{asymptotic} if~$\limsup_{t \to \infty}\dist(\gamma(t), \gamma'(t)) < \infty$.
    It is known that the values of $f^\infty(\gamma)$ and $Q^*(-\dot\gamma)$ do not change among asymptotic geodesic rays \cite{KLM2009,Hirai2024,Hirai2025}.
    Using the fact that any geodesic ray on $\PD$ is asymptotic to one through~$I$, we can replace $\sup_{g \in \GL}$ with $\sup_{u \in \U}$ in the above proof, yielding the following variant of \cref{thm:generaltheorem}:
    \begin{equation}\label{eq:K general}
        \inf_{x \in \PD} Q(\diff f(x)) = \sup_{u \in \U}\inf_{y \in \RR^{n_1} \times \cdots\times \RR^{n_d}} F(\nabla f^u(y)).
    \end{equation}
    For the proof of this variant, we may restrict $g \in \GL$ to a unitary $u \in \U$ in the intermediate equations \eqref{eq:Qstarminus} and \eqref{eq:fg}--\eqref{eq:torusrecession}, some of which simplify in this case.
\end{rem}

\section{A minimax formula for convex optimization on moment polytopes}
\label{sec:moment-poly}
In this section, we apply the result of the previous section to prove a minimax formula for convex optimization on moment polytopes (\cref{thm:moment-polytope-general-q}, stated in full generality as \cref{thm:generalpolytope} below).
See \cite{BFGOWW_FOCS2019} and \cite[\S{}4]{Hirai2025} for more background and detail on the techniques used herein.

Recall that we defined $\GL = \GL(n_1) \times \cdots \times \GL(n_d) \subseteq \GL(N)$.
Let $\pi\colon \GL \to \GL(\VV)$ be a finite-dimensional rational representation, with Lie algebra representation~$\Pi\colon \Lin(n_1) \times \cdots \times \Lin(n_d) \to \Lin(\VV)$, i.e., $\Pi(X) \coloneqq \partial_t\pi(e^{tX})|_{t = 0}$.
We write $g \cdot v \coloneqq \pi(g)v$ for~$g \in \GL$ and~$v \in \VV$, so $\GL \cdot v$ denotes the orbit of~$v$.
We can always equip~$\VV$ with a inner product~$\langle\cdot, \cdot\rangle$ that is invariant under the maximally compact subgroup $\U = \U(n_1) \times \cdots \times \U(n_d) = \GL \cap \U(N)$; we denote the induced norm by $\|v\| \coloneqq \langle v, v\rangle^{1/2}$.

In this setting, the \emph{moment map}~\cite{Ness-Mumford-84} is defined as
\begin{gather}
    \mu \colon \VV \setminus \{0\} \to \rmH,\notag\\
    \tr[\mu(v)X] \coloneqq \sum_{i = 1}^d\tr[(\mu(v))_iX_i] = \frac{\langle v, \Pi(X)v\rangle}{\|v\|^2}\quad (\forall X \in \rmH).\label{eq:momentmap}
\end{gather}
Note that the moment map is more naturally defined on the projective space~$\PP(\VV)$.
It is easy to check that $\mu(u\cdot v) = u\mu(v)u^\dagger$ for any unitary $u \in \U$.
The \emph{moment polytope}~$\Delta(v)$ of $0 \neq v \in \VV$ is defined as the set of spectra of moment map images of points in the orbit closure in projective space:
\begin{align}
    \Delta(v)
&\coloneqq \{\spec(\mu(w)) \mid [w] \in \overline{\GL \cdot [v]} \}\notag\\
&= \overline{\{\spec(\mu(g\cdot v)) \mid g \in \GL\}}\notag\\
&= \overline{\{\spec(\mu(x\cdot v)) \mid x \in \PD\}},\label{eq:polytopebypd}
\end{align}
where we write $[w] \in \PP(\VV)$ for the point in projective space corresponding to a nonzero vector~$0 \neq w \in \VV$, and we denote~$\spec(H) \coloneqq (\spec(H_1), \ldots, \spec(H_d))$ for $H \in \rmH$; the first equality in \cref{eq:polytopebypd} holds by compactness of the orbit closure in projective space and continuity of $\spec(\cdot)$, while the latter equality is immediate from the polar decomposition: it holds that $\spec(\mu(g\cdot v)) = \spec(u\mu(x\cdot v)u^\dagger) = \spec(\mu(x\cdot v))$ for the unique decomposition $g = ux\ (u \in \U, x \in\PD)$.
Remarkably, $\Delta(v)$ is indeed a convex polytope.

\begin{thm}[\cite{Ness-Mumford-84,Kirwan1984}]
For any $0\neq v\in\VV$, $\Delta(v)$ is a bounded convex polytope in $\RR^{n_1}_{\downarrow} \times \cdots \times \RR^{n_d}_{\downarrow}$.
\end{thm}

Let $\rmT \subseteq \GL$ be the set of diagonal matrices, called the maximal torus.
We canonically identify $\rmT \cong (\CC^\times)^{N} \subseteq \GL(N)$.
Then, we can define a moment map and moment polytope exactly as above, but for the group~$\rmT$ in place of~$\GL$.
Concretely, let $\pi_\rmT$ be the restriction of $\pi$ to $\rmT$.
Its Lie algebra representation~$\Pi_\rmT\colon \CC^N \to \Lin(\VV)$ is then given by $\Pi_\rmT(y) = \Pi(\diagmat y)$.
The associated moment map, which we call the \emph{torus moment map}, is given by
\begin{gather*}
    \mu_\rmT\colon \VV\setminus\{0\} \to \RR^N,\\
    \langle\mu_{\rmT}(v), y\rangle
= \frac{\langle v, \Pi_\rmT(y)v\rangle}{\|v\|^2}
= \frac{\langle v, \Pi(\diagmat y)v\rangle}{\|v\|^2} \quad (\forall y \in \RR^N).
\end{gather*}
By comparing it to \cref{eq:momentmap}, it follows that $\mu_{\rmT}(v) = \diagvec\mu(v)$.
The corresponding moment polytope is given by
\begin{align}
    \Omega(v)
&\coloneqq \{\mu_{\rmT}(w) \mid [w] \in \overline{\rmT \cdot [v]} \}\notag\\
&= \overline{\{\mu_{\rmT}(g\cdot v) \mid g \in \rmT\}}\notag\\
&= \overline{\{\mu_{\rmT}(e^{\diagmat y} \cdot v) \mid y \in \RR^N\}}.\label{eq:support poly}
\end{align}
and we call it the \emph{support polytope} of~$v$ for the following reason:
Because~$\rmT$ is commutative, we can jointly diagonalize the action of~$\rmT$ and decompose~$\VV = \VV_{\omega_1} \oplus \cdots \oplus \VV_{\omega_r}$ into subspaces~$\VV_{\omega_\ell}$, indexed by vectors $\omega_\ell \in \ZZ^N$ called~\emph{weights}, such that $\rmT$ acts as~$g \cdot u = e^{\braket{\omega_\ell, z}} u$ for all~$g = e^z \in (\CC^\times)^N$ and~$u \in \VV_{\omega_\ell}$.
Any vector~$v \in \VV$ admits a unique decomposition~$v = v_{\omega_1} + \cdots + v_{\omega_r}$ ($v_{\omega_\ell} \in \VV_{\omega_\ell}$); we define its~\emph{support}~as
\[
    \supp(v) \coloneqq \{\omega_\ell \mid v_{\omega_\ell} \neq 0\}.
\]
Then the moment map can be computed explicitly as $\mu_\rmT(v) = \sum_{\ell = 1}^r\frac{\|v_{\omega_\ell}\|^2}{\|v\|^2}\omega_\ell \in \conv\supp(v)$, and therefore $\Omega(v) \subseteq \conv\supp(v)$.
The opposite inclusion also holds (see, e.g., \cite{BLNW2020}; cf.\ \cite{Atiyah1982,BFGOWW_FOCS2019}), and hence the support polytope is described as follows, justifying the terminology:
\begin{equation}\label{eq:support poly concrete}
    \Omega(v) = \conv\supp(v).
\end{equation}
For a fixed representation~$\VV$, there are only finitely many supports and hence there are only finitely many possible support polytopes.
In particular, for any fixed~$v$, there are only finitely many support polytopes $\Omega(g \cdot v)$ as we vary over~$g \in \GL$.

\begin{exa}[Tensor action]
    Consider the complex vector space $\VV \coloneqq \CC^{n_1} \otimes \cdots \otimes \CC^{n_d} \cong \CC^{n_1 \times \cdots \times n_d}$.
    Its element~$t$ is described by complex numbers $(t_{j_1, \ldots, j_d}) \in \CC$ indexed by~$1 \le j_i \le n_i$.
    Then, the \emph{tensor action} of $\GL$ on $\VV$ is defined by
    \begin{align*}
        (g_1, \cdots, g_d)\cdot t &\coloneqq (g_1 \otimes \cdots \otimes g_d)t
        \\
        &\qquad(g_i \in \GL(n_i), t \in \CC^{n_1} \otimes \cdots \otimes \CC^{n_d}).
    \end{align*}
    The standard $\ell^2$-inner product $\displaystyle\langle t, s\rangle \coloneqq \sum_{j_1, \ldots, j_d} \bar t_{j_1, \ldots, j_d} s_{j_1, \ldots, j_d}$ is unitarily invariant.
    The Lie algebra representation for the tensor action is given by
    \begin{align*}
      \Pi(X_1,\dots,X_d) t
    &= \sum_{i=1}^d \left( I^{\ot (i-1)} \ot X_i \ot I^{\ot (d-i)} \right)
    t\\
    &\qquad(X_i \in \Lin(n_i), t \in \CC^{n_1} \otimes \cdots \otimes \CC^{n_d}).
    \end{align*}
    We can view a tensor~$t$ as a linear map $t_i \colon \CC^{n_i} \to \bigotimes_{j \neq i}\CC^{n_j}$, which is called the \emph{$i$-th flattening of $t$}.
    Then the moment map~\eqref{eq:momentmap} computes the normalized \emph{one-body quantum marginals}~$\rho_i(t) \coloneqq {t_i^\dagger t_i}/{\|t\|^2}$~of~$t$:
    \[
        \mu(t)_i
        = \frac{t_i^\dagger t_i}{\|t\|^2}
        \qquad (i \in [d]).
    \]
    Thus, the moment polytope of a tensor~$t$, which is also known as its \emph{entanglement polytope}~\cite{walter2013entanglement}, are given by~\eqref{eq:entanglement polytope intro}.
\footnote{Because the tensor action contains the dilations, we can replace the closure in projective space by (the nonzero vectors in) the closure in the vector space.}
    Since the action of $\rmT$ is diagonal with respect to the standard basis vectors~$e_{j_1} \ot \cdots \ot e_{j_d}$, with weights $\omega = (e_{j_1},\dots,e_{j_d})$, the support polytope is given by~\eqref{eq:support polytope intro}.
\end{exa}

\begin{exa}[Diagonal action]\label{ex:diagonal-action}
The \emph{diagonal action} is obtained by restricting the tensor action on~$\VV \coloneqq (\CC^n)^{\otimes d}$ to the subgroup $\GL(n) \subseteq \GL(n)^d$. That is:
\[
    g\cdot t \coloneqq g^{\otimes d}t \qquad (g \in \GL(n), t \in (\CC^n)^{\otimes d}).
\]
The moment map for this action is given by the sum of the quantum marginals, and the moment polytope contains the corresponding spectra.
Denoting the former by~$\mu_{\Sym}$ and the latter by~$\Delta_{\Sym}$ to avoid confusion with the tensor action, we have
\begin{align}
    \mu_{\Sym}(t) &= \frac1{\|t\|^2}\sum_{i = 1}^dt_i^\dagger t_i,\notag\\
    \Delta_{\Sym}(t) &= \left\{\spec\left(\frac1{\|s\|^2}\sum_{i = 1}^ds_i^\dagger s_i\right)\ \middle|\ 0\neq s \in \overline{\GL(n) \cdot t} \right\}.\label{eq:diagonal-polytope}
\end{align}
The support polytope, which we will similarly denote by~$\Omega_{\Sym}$, is given by
\begin{align}
    \Omega_{\Sym}(t)
    &= \conv \left\{e_{j_1} + \cdots + e_{j_d} \in \RR^n \mid t_{j_1, \ldots, j_d} \neq 0\right\}\notag\\
    &= \left\{ {\textstyle\sum}_{i=1}^d p_i \mid p \in \Omega(t) \right\}\label{eq:diagonal support polytope}
\end{align}
because for the diagonal action the weight of a standard basis vector~$e_{j_1} \ot \cdots \ot e_{j_d}$ is $e_{j_1} + \cdots + e_{j_d}$.

The action of $\GL(n)$ on $(\CC^n)^{\otimes d}$ is not irreducible.
Important subrepresentations are given by the spaces of \emph{symmetric tensors} and \emph{anti-symmetric tensors}
\begin{align*}
    \Sym^d(\CC^n) &\coloneqq \{t \in (\CC^n)^{\otimes d} \mid t_{\sigma(j_1), \ldots, \sigma(j_d)} = t_{j_1, \ldots, j_d} \; \forall \sigma \in S_d\},\\
    \Alt^d(\CC^n) &\coloneqq\\
    \{t \in (&\CC^n)^{\otimes d} \mid t_{\sigma(j_1), \ldots, \sigma(j_d)} = \mathrm{sign}(\sigma) \, t_{j_1, \ldots, j_d} \; \forall \sigma \in S_d \}.
\end{align*}
    Since all of the $d$ marginals of a symmetric or anti-symmetric tensor are the same, the moment map and the moment polytope of such tensors can be simply written as
    \begin{align*}
        \mu_{\Sym}(t) &= d\frac{t_1^\dagger t_1}{\|t\|^2},\\
        \Delta_{\Sym}(t) &= d \left\{ \spec\frac{s_1^\dagger s_1}{\|s\|^2}\ \middle|\ 0\neq s \in \overline{\GL(n) \cdot t} \right\}\\
        &\qquad\qquad\qquad(0\neq t \in \Sym^d(\CC^n) \cup \Alt^d(\CC^n)).
    \end{align*}
\end{exa}

\bigskip

We now discuss minimization problems on general moment and support polytopes and how they can be related.
Fix $0 \neq v \in \VV$ and let $F \colon \RR^{n_1} \times \cdots\times \RR^{n_d} \to \RR \cup \{\infty\}$ be a function that satisfies the following conditions (the labels match the ones in the previous section):
\begin{enumerate}[leftmargin=4em,noitemsep,label={(F4)}]
    \item[\ref{it:F1}] $F$ is convex and l.s.c.
    \item[\ref{it:F2}] $F$ is symmetric in each argument.
    \item\label{it:F4} $\Delta(v) \subseteq \dom F$.
\end{enumerate}
Recall that the symmetry and convexity of $F$ imply~$F(\diagvec X) \le F(\spec X)$ for every~$X \in \rmH$.
Considering~$X \coloneqq \mu(v)$ and using $\mu_{\rmT}(v) = \diagvec\mu(v)$, we have~$F(\mu_{\rmT}(v)) \le F(\spec\mu(v))$ for all $v \in \VV$.
Therefore, \ref{it:F4} implies that~$\Omega(g\cdot v) \subseteq \dom F$ for every~$g \in \GL$.
We now show that the minimum of~$F$ over the moment polytope~$\Delta(v)$ equals the maximum over~$g \in \GL$ of the minimum of~$F$ over the support polytopes~$\Omega(g \cdot v)$, as a consequence of~\cref{thm:generaltheorem}.

\begin{thm}[Detailed version of \cref{thm:moment-polytope-general-q}]\label{thm:generalpolytope}
    Let $\pi\colon \GL \to \GL(\VV)$ be a finite-dimensional rational representation, and let $0 \neq v \in \VV$.
    Let~$F\colon \RR^{n_1} \times \cdots \times \RR^{n_d} \to \RR \cup \{\infty\}$ be a function satisfying \ref{it:F1}, \ref{it:F2}, and \ref{it:F4}.
    Then,
    \[
        \min_{p \in \Delta(v)}F(p) = \max_{g \in \GL}\min_{p \in \Omega(g\cdot v)}F(p).
    \]
\end{thm}

Before giving the proof of the theorem, we note that to prove the theorem the value of~$F$ outside of the moment polytope and the (finitely many possible) support polytopes does not matter.
As these are bounded convex polytopes, we may assume without loss of generality that $F(p) = \infty$ for sufficiently large~$\|p\|$ and hence that
\begin{itemize}[leftmargin=4em]
    \item[\ref{it:F3}] $F$ has bounded sublevel sets.
\end{itemize}
The key idea of the reduction from \cref{thm:generalpolytope} to \cref{thm:generaltheorem} is the well-known fact that the moment map can interpreted as the gradient of a geodesically convex function, called the \emph{Kempf--Ness function}~\cite{Ness-Mumford-84,Kirwan1984} (see also \cite{BFGOWW_FOCS2019,HiraiSakabe2024,Hirai2025}).
In the present setting, the Kempf--Ness function of a vector~$0 \neq v \in \VV$ is defined as
\begin{align*}
    f_v \colon \PD \to \RR,
    \quad
    f_v(x) \coloneqq \log\langle v, x\cdot v\rangle.
\end{align*}
Then the following properties are well-known:
\begin{lem}[{e.g., \cite{BFGOWW_FOCS2019}, \cite{HiraiSakabe2024}}]
\label{prop:kempf-ness}
Let $0\neq v \in \VV$. Then:
    \begin{enumerate}
        \item[(a)] The Kempf--Ness function $f_v$ is geodesically convex.
        \item[(b)] The differential of~$f_v$ and the moment map are related as follows:
        for every $x\in\PD$,
        \[
            \tau^*_{I \to x}\diff f_v(x) = x^{1/2}\diff f_v(x)x^{1/2} = \mu(\pi(x^{1/2})v),
        \]
        where for the first equality we identify $T^*_x\PD \cong \rmH$.
    \end{enumerate}
\end{lem}

\begin{proof}[Proof of \cref{thm:generalpolytope}]
By \cref{eq:dfn_QF} and \cref{prop:kempf-ness}~(b), it holds that, for any $x \in \PD$,
\begin{align}
    Q(\diff f_v(x))
&= F_\rmH\mleft(\tau_{I \to x}^*\diff f_v(x)\mright)
= F_{\rmH}\mleft(\mu\mleft(\pi\mleft(x^{1/2}\mright)v\mright)\mright)\notag\\
&= F\mleft(\spec\mu\mleft(\pi\mleft(x^{1/2}\mright)v\mright)\mright).\label{eq:qf_to_f}
\end{align}
Therefore, by \cref{eq:polytopebypd} and the fact that $F$ is continuous on~$\Delta(v)$ (\cite[Thm.~10.2]{Rockafellar1970}), we have
\begin{equation}\label{eq:lhs for moment polytopes}
    \inf_{x \in \PD}Q(\diff f_v(x)) = \min_{p \in \Delta(v)} F(p).
\end{equation}
The continuity assumption~\ref{it:C} is satisfied in this case, as follows.
Since $\spec: \rmH \to \RR^N$ is continuous, the continuity of $F$ on $\Delta(v)$ implies continuity of $F_\rmH$ on the set $\U\cdot\Delta(v) \coloneqq \{u\diagmat(p)u^\dagger \mid u \in \U, p \in \Delta(v)\}$; hence, $F_\rmH$ is continuous on $\overline{\{\tau^*_{I \to x}\diff f_v(x) \mid x \in \PD\}} \subseteq \U\cdot\Delta(v)$, where the inclusion follows from \cref{eq:dfn_QF} and \cref{prop:kempf-ness}~(b).

On the other hand, in this setting, the function~$f_v^g$ defined as in \cref{eq:fg,eq:fg alternative} for~$f$ the Kempf--Ness function~$f_v$, is given by
\[
    f_v^g:\RR^N \to \RR,\ f_v^g(y) = \log\left\langle g\cdot v, e^{\diagmat y}\cdot g\cdot v\right\rangle \eqqcolon f^\rmT_{g\cdot v}(y).
\]
We have introduced the notation $f^\rmT_{g\cdot v} \colon \RR^N \to \RR$ because this function is the Kempf--Ness function for the torus action.
Concretely, it is a log-sum-exp function as in \emph{geometric programming}~\cite{BKVH2007,BLNW2020}.
By a simple calculation (or invoking \cref{prop:kempf-ness}~(b) for $\rmT$), its gradient is given by the torus~moment~map:
\[
    \nabla f^\rmT_{g\cdot v}(y) = \mu_\rmT\left(e^{\diagmat y}\cdot g\cdot v\right).
\]
Therefore, by \cref{eq:support poly}, we have
\[
    \inf_{y \in \RR^N}F(\nabla f^g_v(y)) = \inf_{y \in \RR^N}F\left(\nabla f^\rmT_{g\cdot v}(y)\right) = \min_{p \in \Omega(g\cdot v)} F(p).
\]
Together with \cref{eq:lhs for moment polytopes} we see that \cref{thm:generaltheorem}, applied to the Kempf--Ness function~$f_v$, reads
\[
    \min_{p \in \Delta(v)}F(p) = \sup_{g \in \GL}\min_{p \in \Omega(g\cdot v)}F(p).
\]
The above supremum is always attained, as there are only finitely many possible support polytopes.
\end{proof}

\begin{rem}\label{rem:K general mopo}
    If we use formula~\eqref{eq:K general} in \cref{rem:K_general} in place of \cref{thm:generaltheorem}, we obtain the following variant of \cref{thm:generalpolytope}:
    \[
        \min_{p \in \Delta(v)}F(p) = \max_{u \in \U}\min_{p \in \Omega(u\cdot v)}F(p).
    \]
\end{rem}

\section{Applications and connections}\label{sec:applications}
In this section we describe concrete applications of \cref{thm:generaltheorem,thm:generalpolytope} to a variety of tensor parameters.
In \cref{sub:quantum functionals}, we discuss the most important application, the quantum functionals, and present a simple proof of being universal spectral points.
In \cref{sub:symmetric quantum functionals}, we discuss results for the symmetric quantum functional analogous to the (usual) quantum functionals.
In \cref{sub:asymptotic slice rank}, we discuss the asymptotic weighted slice ranks and show the relation to the asymptotic weighted vertex cover number (\cref{thm:intro-weighted-sr}).
In \cref{sub:G stable}, we explain how the characterization of Derksen's $G$-stable rank in terms of entanglement polytopes can be obtained from our theorem (\cref{prop:fractional-cover-reciprocal}).
Finally, in \cref{sub:ncrank}, we relate formulas for the non-commutative rank due to Hirai and Fortin--Reutenauer.

\subsection{Quantum functionals}\label{sub:quantum functionals}
\cref{thm:support-equals-quantum} provides an immediate and conceptually simple proof that the quantum functionals~$F_\theta$ are universal spectral points.
The important observation is that properties required by Strassen's axioms follow by combining elementary properties of the moment polytope $\Delta(t)$ and the support polytope $\Omega(t)$, combined with the key equality provided by our \cref{thm:support-equals-quantum}:
\[
F_{\theta}(t) = \zeta^\theta(t)
\]
Throughout, we let $0\neq s\in\CC^{m_1}\otimes\dots\otimes\CC^{m_d}$ and $0\neq t\in\CC^{n_1}\otimes\dots\otimes\CC^{n_d}$.
We denote by $\GL(\mathbf{m})$ and $\GL(\mathbf{n})$ the products of general linear groups that act on $s$ and $t$, respectively.

The following properties of the Shannon entropy are well-known:

\begin{lem}[Properties of the Shannon entropy]
\label{lem:additivity and recursivity}
Let $q\in\RR^m$ and $r\in\RR^n$ be two probability distributions.
Then:
\begin{enumerate}[label={(\alph*)}]
\item\label{it:add} \emph{Additivity with respect to tensor products:} $H(q \ot r) = H(q) + H(r)$
\item\label{it:subadd} \emph{Subadditivity:} $H(p) \leq H(q) + H(r)$ if $q,r$ are the marginal distributions of a joint distribution~$p$.
\item\label{it:chain} \emph{Recursivity (or chain rule):}
$H\mleft( \lambda q \oplus (1-\lambda) r \mright) = \lambda H(q)+ (1-\lambda) H(r) + h(\lambda)$ for every $\lambda \in [0,1]$, where the last term is the binary Shannon entropy $h(\lambda)\coloneqq -\lambda\log_2(\lambda) -(1-\lambda)\log_2 (1-\lambda)$.
\end{enumerate}
\end{lem}

Furthermore, we have the following elementary equality.

\begin{lem}
\label{lem:strassen-lemma}
For $X,Y\in\RR$, it holds that
    \[\max_{\lambda \in [0,1]} 2^{\lambda X + (1-\lambda)Y+h(\lambda)} = 2^X + 2^Y.\]
\end{lem}
\begin{proof}
    This directly follows from differentiating with respect to~$\lambda$ and using the concavity of~$h$.
\end{proof}

We now recall the well-known fact (and proof) that~$F_\theta$ is super-additive and super-multiplicative.

\begin{lem}
\label{lem:quantum-sups}
    For all $s, t\neq0$ and $\theta\in\Theta$, we have
    $F_\theta(s\otimes t)\geq F_\theta(s) F_\theta(t)$
    and
    $F_\theta(s\oplus t) \geq F_\theta(s)+F_\theta(t)$.
\end{lem}
\begin{proof}
For $i\in [d]$, the $i$-th flattening $(s\otimes t)_i$ of $s\otimes t$ is~$s_i \otimes t_i$.
Hence, for $g\in\GL(\mathbf{m})$ and $h\in\GL(\mathbf{n})$ we have~$g \ot h \in \GL(\mathbf{m} \cdot \mathbf{n})$ and
$\mu((g \ot h) \cdot (s \ot t)) = \mu((g\cdot s)\otimes (h\cdot t)) = \mu(g\cdot s)\otimes \mu(h\cdot t)$, where $\mu$ denotes the moment map.
As a consequence, the moment polytope of a tensor product satisfies
\begin{align*}
    \Delta(s \ot t) \supseteq \left\{ (q \ot r)_\downarrow \mid q \in \Delta(s), r \in \Delta(t) \right\},
\end{align*}
where $(\cdot)_{\downarrow}$ denotes the sorting of each part in decreasing order.
Consider any $p = (q \ot r)_{\downarrow} \in \Delta(s \ot t)$, where $q \in \Delta(s), r \in \Delta(t)$.
Using the additivity property of the Shannon entropy (\cref{lem:additivity and recursivity},~\ref{it:add}), and taking the desired convex combination, we find that
$\sum_i \theta_i H(p_i) = \sum_i \theta_i H(q_i) + \sum_i \theta_i H(r_i)$,
and hence we obtain super-multiplicativity:
\begin{align*}
    F_\theta(s\otimes t) &= \max_{p \in \Delta(s \ot t)} 2^{\sum_i \theta_i H(p_i)}\\
    &\geq \max_{q \in \Delta(s), r \in \Delta(t)} 2^{\sum_i \theta_i H(q_i) + \sum_i \theta_i H(r_i)}\\
    &= F_\theta(s) F_\theta(t)
\end{align*}

Similarly, for $i\in [d]$, the $i$-th flattening $(s \op t)_i$ of $s \op t$ is~$s_i  \op t_i$.
Hence, for $g\in\GL(\mathbf{m})$ and $h\in\GL(\mathbf{n})$ we have~$(g \op h)\in \GL(\mathbf{m} + \mathbf{n})$ and
$\mu((g \op h) \cdot (s \op t)) = \mu((g\cdot s) \op (h\cdot t)) = \lambda \mu(g\cdot s) \op (1-\lambda) \mu(h\cdot t)$, where $\lambda = \frac {\norm{g\cdot s}^2} {\norm{g\cdot s}^2 + \norm{h\cdot t}^2}$.
As a consequence, the moment polytope of a direct sum satisfies
\[
    \Delta(s \oplus t)\supseteq\{ (\lambda q \oplus (1-\lambda)r)_{\downarrow} \mid \lambda \in [0,1], q\in \Delta(s), r\in \Delta(t)\},
\]
where $(\cdot)_{\downarrow}$ denotes the sorting each part in a decreasing order.
Consider an arbitrary point of the form $p=(\lambda q \oplus (1-\lambda) r)_{\downarrow}\in\Delta(s\oplus t)$, where $q\in\Delta(s), r\in\Delta(t)$.
Using the recursivity property of the Shannon entropy (\cref{lem:additivity and recursivity},~\ref{it:chain}), we taking the desired convex combination, we find that
\[
    \sum_i \theta_i H(p_i) = \lambda \sum_i \theta_i H(q_i) + (1-\lambda) \sum_i \theta_i H(r_i) + h(\lambda).
\]
Similarly as above, we obtain super-additivity:
\begin{align*}
    F_\theta(s\oplus t)&=\max_{p\in\Delta(s\oplus t)} 2^{\sum_i \theta_i H(p_i)}\\
    &\geq \max_{\substack{q\in\Delta(s)\\r\in\Delta(t)}} \max_{\lambda\in [0,1]} 2^{\lambda \sum_i \theta_i H(q_i)+(1-\lambda)\sum_i \theta_i H(r_i)+h(\lambda)}\\
    &=F_\theta(s)+F_\theta(t),
\end{align*}
where we use \cref{lem:strassen-lemma} in the last equality to carry out the maximization over~$\lambda$.
\end{proof}

We now prove that the support functional~$\zeta^\theta$ is sub-additive and sub-multiplicative.

\begin{lem}
\label{lem:support-subs}
    For all $s, t\neq 0$ and $\theta\in\Theta$, we have $\zeta^\theta(s\otimes t)\leq \zeta^\theta(s)\zeta^\theta(t)$ and $\zeta^\theta(s\oplus t)\leq \zeta^\theta(s)+\zeta^\theta(t)$.
\end{lem}
\begin{proof}
    By \cref{eq:support polytope intro}, we have
    \begin{align*}\Omega(s) &= \conv\left\{ (e_{i_1},\dots,e_{i_d}) : s_{i_1,\dots,i_d} \neq 0 \right\},\\
    \Omega(t) &= \conv\left\{ (e_{j_1},\dots,e_{j_d}) : t_{j_1,\dots,j_d} \neq 0 \right\}.\end{align*}
On the other hand, identifying~$\RR^{mn} \cong \RR^m \ot \RR^n$, we have
\begin{align*}
    \Omega(s\otimes t) = \conv\{&(e_{i_1} \ot e_{j_1}, e_{i_2} \ot e_{j_2}, \dots, e_{i_d} \ot e_{j_d} )\\
    &\mid s_{i_1,i_2,\dots,i_d}\neq 0 \text{ and } t_{j_1,j_2,\dots,j_d}\neq 0\}.
\end{align*}
Hence, by the bilinearity of the tensor product, \begin{align}
&\Omega(s\otimes t) =\notag\\
&\conv\{ (q_1\otimes r_1,q_2\otimes r_2,\dots,q_d\otimes r_d)\mid q\in\Omega(s), r\in\Omega(t) \},\label{eq:support-poly-tensor}
\end{align} which holds for any tensors $s$ and $t$.
For $p \in \Omega(s\otimes t)$, we can interpret each component~$p_i$ as a joint probability distribution on $[m_i] \times [n_i]$.
Then \cref{eq:support-poly-tensor} implies that there exist $q \in \Omega(s)$ and $r \in \Omega(t)$ such that, for all $i\in[d]$, $q_i$ and $r_i$ are the marginal distributions of~$p_i$.
By the sub-additivity of the Shannon entropy (\cref{lem:additivity and recursivity},~\ref{it:subadd}), we have $\sum_i \theta_i H(p_i) \leq \sum_i \theta_i H(q_i)+ \sum_i \theta_i H(r_i)$, and hence
\begin{align}\label{eq:intermediate}
    \max_{p\in \Omega(s\otimes t)} 2^{\sum_i\theta_i H(p_i)}
\leq \max_{q\in \Omega(s), r \in \Omega(t)} 2^{\sum_i\theta_i H(q_i) + \sum_i\theta_i H(r_i)}.
\end{align}
(In fact, equality holds as a consequence of~\eqref{eq:support-poly-tensor}, but we will not need this.)
Thus we find that the support functional is sub-multiplicative:
\begin{align*}
\zeta^\theta(s\otimes t)&=\min_{G\in \GL(\mathbf{m} \cdot \mathbf{n})} \max_{p\in \Omega(G\cdot (s\otimes t))} 2^{\sum_i\theta_i H(p_i)}\\
&\leq \min_{\substack{g\in \GL(\mathbf{m})\\ h\in \GL(\mathbf{n})}} \max_{\substack{q\in\Omega(g\cdot s)\\ r\in\Omega(h\cdot t)}} 2^{\sum_i \theta_i H(q_i) + \sum_i \theta_i H(r_i)}\\
&= \zeta^\theta(s)\zeta^\theta(t),
\end{align*}
where we obtain the inequality by restricting the minimization to group elements of the form~$G = g \ot h$ and using \cref{eq:intermediate} with~$g \cdot s$ and~$h \cdot t$ in place of~$s$ and~$t$.

We now prove sub-additivity.
Because the support of $s \op t$ is the disjoint union of the supports of~$s$ and~$t$, it follows at once that
\begin{align*}
    \Omega(s \op t) = \{ \lambda q \op (1-\lambda) r \mid \lambda \in [0,1], q \in \Omega(s), r \in \Omega(t) \}.
\end{align*}
Hence,
\begin{align}
  \max_{p\in \Omega(s \op t)} 2^{\sum_i \theta_i H(p_i)}
&= \max_{\substack{q\in\Omega(s) \\ r\in\Omega(t)}} \max_{\lambda \in [0,1]} 2^{\sum_i \theta_i H(\lambda q_i \oplus (1-\lambda)r_i) }\notag\\
&= \max_{\substack{q\in\Omega(s) \\ r\in\Omega(t)}} 2^{\sum_i \theta_i H(q_i)} + 2^{\sum_i \theta_i H(r_i)},\label{eq:intermediate 2}
\end{align}
where we used \cref{lem:additivity and recursivity}~\ref{it:chain} and \cref{lem:strassen-lemma} to obtain the second equality.
Thus we find that the support functional is also sub-additive:
\begin{align*}
    \zeta^\theta(s \op t)
    &= \min_{G\in \GL(\mathbf{m} + \mathbf{n})} \max_{p\in \Omega(G\cdot (s\otimes t))} 2^{\sum_i\theta_i H(p_i)}\\
    &\leq \min_{\substack{g\in \GL(\mathbf{m})\\ h\in \GL(\mathbf{n})}} \max_{\substack{q\in\Omega(g\cdot s)\\ r\in\Omega(h\cdot t)}} 2^{\sum_i \theta_i H(q_i)} + 2^{\sum_i \theta_i H(r_i)}\\
    &= \zeta^\theta(s) + \zeta^\theta(t),
\end{align*}
where we obtain the inequality by restricting the minimization to group elements of the form~$G = g \op h$ and using \cref{eq:intermediate 2} with~$g \cdot s$ and~$h \cdot t$ in place of~$s$ and~$t$.
\end{proof}

The equality $F_\theta(t)=\zeta^\theta(t)$ now yields:

\begin{cor}
    For any $\theta\in\Theta$, the quantum functional $F_\theta$ is a universal spectral point.
\end{cor}
\begin{proof}
    It is easy to see from the definition that $F_\theta$ is normalized on the unit tensor and monotone under restriction (because the moment polytope is monotone under restriction).
    By \cref{lem:quantum-sups}, $F_\theta$ is super-additive and super-multiplicative.
    By \cref{lem:support-subs}, $\zeta^\theta$ is sub-additive and sub-multiplicative.
    By \cref{thm:support-equals-quantum}, $\zeta^\theta = F_\theta$.
    We conclude that $F_\theta$ is a universal spectral point.
\end{proof}

\subsection{Symmetric quantum functionals}\label{sub:symmetric quantum functionals}
In \cref{sec:introduction} we proved the equality of the quantum functionals and Strassen's support functionals (\cref{thm:support-equals-quantum} in the introduction) and explained how this equality gives a new direct proof of the fact that these functionals are spectral points.
In this section we discuss analogous results for \emph{symmetric quantum functional}---a variant of the quantum functionals constructed for the diagonal action of~$\GL(n)$ on~$(\CC^n)^{\ot d}$ introduced in~\cite{CFTZ-22}.
The symmetric quantum functional is well-defined for arbitrary tensors and known to be a spectral point in the asymptotic spectrum of symmetric tensors.
Moreover, it is sub-multiplicative for all tensors.
Here we provide a minimax description of the symmetric quantum functional, in the spirit of Strassen's support functionals, and provide a new direct proof of its sub-multiplicativity.

Let $\CC^n$ be the $n$-dimensional vector space, and consider the diagonal action of $\GL(n)$ on $(\CC^n)^{\otimes d}$, as in \cref{ex:diagonal-action}.
For a tensor $0 \neq t\in(\CC^n)^{\otimes d}$, we denote by $\Delta_{\Sym}(t)\subseteq\RR^n$ the moment polytope, and by $\Omega_{\Sym}(t)$ the support polytope with respect to this action, see \cref{eq:diagonal-polytope,eq:diagonal support polytope}.
We reserve the notation~$\Delta(t)$ (resp. $\Omega(t)$) for the moment polytope (resp. support polytope) of $t$ with respect to the action of the larger group $\GL\coloneqq \GL(n)^{\times d}$.

\begin{defn}[Symmetric quantum functional]
The \emph{symmetric quantum functional} of a tensor $t \in (\CC^n)^{\ot d}$ is defined as
\[
    F_{\Sym}(t) \coloneqq \max_{p\in\Delta_{\Sym}(t)}\, 2^{H(p/d)},
    \] where $H$ denotes the Shannon entropy.
\end{defn}

\noindent
We note that the elements of the moment polytope $\Delta_{\Sym}(t)$ are vectors $p\in\RR^{n}_{\geq0}$ with $\sum_{i=1}^n p_i=d$, so the division by $d$ results in a probability distribution.
Since $p \mapsto -H(p/d)$ is a convex symmetric l.s.c.\ function, a straightforward application of \cref{thm:moment-polytope-general-q} yields the following new description:

\begin{restatedtheorem}{cor:symmetric-functional}
For every tensor $0 \neq t \in (\CC^n)^{\otimes d}$, we have
\[
    F_{\Sym}(t) = \min_{g\in\GL(n)} \max_{p\in\Omega_{\Sym}(g \cdot t)} 2^{H(p/d)}.
\]
\end{restatedtheorem}

It was proved in~\cite{CFTZ-22} that $F_{\Sym}$ is sub-multiplicative, i.e., for $0 \neq t\in(\CC^n)^{\otimes d}$ and $0 \neq s\in(\CC^m)^{\otimes d}$, we have $F_{\Sym}(t\otimes s)\leq F_{\Sym}(t)F_{\Sym}(s)$.
The proof relies on a non-trivial relationship between the moment polytopes~$\Delta_{\Sym}(t\otimes s)$ and~$\Delta_{\Sym}(t), \Delta_{\Sym}(s)$.
More concretely, they proved that the former polytope is contained in a certain product of the latter polytopes, using a representation-theoretic description of the moment polytopes.
Here we reprove the sub-multiplicativity of~$F_{\Sym}$ using \cref{cor:symmetric-functional} and a simple description of the support polytope~$\Omega_{\Sym}(t\otimes s)$ in terms of $\Omega(t)$ and $\Omega(s)$ provided by the following~lemma:

\begin{lem}
\label{lem:product-support}
    Let $0 \neq t\in(\CC^n)^{\otimes d}$ and $0 \neq s\in(\CC^m)^{\otimes d}$ be $d$-tensors.
    Then,
    \begin{equation}\label{eq:omega sym t s}
    \Omega_{\Sym}(t\otimes s) = \conv\left\{ \sum_{i=1}^d q_i \ot r_i \mid q\in \Omega(t), r\in \Omega(s) \right\}.
    \end{equation}
\end{lem}
\begin{proof}
The claim follows from \cref{eq:support-poly-tensor} using $\Omega_{\Sym}(t\otimes s) = \{\sum_{i=1}^d p_i\mid p\in \Omega(t\otimes s)\}$ (see \cref{eq:diagonal support polytope}) and the linearity of the map $p\mapsto \sum_i p_i$.
\end{proof}

\begin{cor}\label{cor:product marginals}
    Let $0 \neq t\in(\CC^n)^{\otimes d}$ and $0 \neq s\in(\CC^m)^{\otimes d}$ be $d$-tensors.
    Let $p \in \Omega_{\Sym}(t\otimes s) \subseteq \RR^{nm} \cong \RR^n \otimes \RR^m$.
    Interpreting~$p/d$ as a joint probability distribution on $[n] \times [m]$, its marginal distributions are of the form~$q/d$ and~$r/d$, where~$q \in \Omega_{\Sym}(t)$ and $r \in \Omega_{\Sym}(s)$.
\end{cor}
\begin{proof}
This follows from \cref{eq:omega sym t s,eq:diagonal support polytope}, by observing that~$\Omega(t)$ and~$\Omega(s)$ consist of $d$-tuples of probability distributions.
\end{proof}

\begin{prop}[{\cite[Prop.~B.2]{CFTZ-22}}]
\label{prop:symmetric-sub-mult}
    For $0 \neq t\in(\CC^n)^{\otimes d}$ and $0 \neq s\in(\CC^m)^{\otimes d}$, we have $F_{\Sym}(t\otimes s)\leq F_{\Sym}(t) F_{\Sym}(s)$.
\end{prop}
\begin{proof}
By \cref{cor:symmetric-functional}, we have
\begin{align*}
    F_{\Sym}(t \ot s)
    &= \min_{G\in\GL(nm)} \max_{p\in\Omega_{\Sym}(G \cdot (t \ot s))} 2^{H(p/d)}\\
    &\leq \min_{\substack{g\in\GL(n)\\ h\in\GL(m)}} \max_{p\in\Omega_{\Sym}((g \cdot t) \ot (h \cdot s))} 2^{H(p/d)}
\end{align*}
By \cref{cor:product marginals}, for any point~$p\in\Omega_{\Sym}((g \cdot t) \ot (h \cdot s))$, we know that~$p/d$ is a joint distribution with marginal distributions~$q/d$ and~$r/d$, where $q \in \Omega_{\Sym}(g \cdot t)$ and $r \in \Omega_{\Sym}(h \cdot s)$.
Since the Shannon entropy is sub-additive (\cref{lem:additivity and recursivity},~\ref{it:subadd}), we have~$H(p/d) \leq H(q/d) + H(r/d)$, and hence
\begin{align*}
    &\min_{\substack{g\in\GL(n)\\ h\in\GL(m)}} \max_{p\in\Omega_{\Sym}((g \cdot t) \ot (h \cdot s))} 2^{H(p/d)}\\
    &\quad\leq \min_{\substack{g\in\GL(n)\\ h\in\GL(m)}} \max_{\substack{q\in\Omega_{\Sym}(g \cdot t)\\ r\in\Omega_{\Sym}(h \cdot s)}} 2^{H(q/d) + H(r/d)}
= F_{\Sym}(t) F_{\Sym}(s).
\qedhere
\end{align*}
\end{proof}

\subsection{Asymptotic slice rank}\label{sub:asymptotic slice rank}
To any tensor $t\in\CC^{n_1}\otimes\CC^{n_2}\otimes\dots\otimes\CC^{n_d}$, one can assign a $d$-uniform, $d$-partite hypergraph~$\HH_t$, whose vertex set is $[n_1]\sqcup [n_2] \sqcup\dots\sqcup [n_d]$ and the hyperedge set~$E(\HH_t)$ is determined by the tensor's support, i.e., there is an edge $(i_1,i_2,\dots,i_d)$ if and only if $t_{i_1,i_2,\dots,i_d}\neq 0$.
For a hypergraph invariant~$\iota$, the quantity
\[
\min_{g\in\GL} \iota(\HH_{g\cdot t})
\]
can then be interpreted as a notion of ``rank'' for tensors.
In this section we discuss the slice rank, which corresponds the choice where $\iota$ is the vertex cover number.
Other notable examples include the $G$-stable rank and the non-commutative rank, which will be discussed in subsequent sections.

We first give a geometric definition of the (weighted) slice rank.
The original definition is due to Tao~\cite{Tao-16,Tao-Sawin-16}. Its weighted generalization was given in~\cite{CLZ-23}.

\begin{defn}[Slice decomposition]
We say that a tensor~$t\in\CC^{n_1}\otimes\CC^{n_2}\otimes\dots\otimes\CC^{n_d}$ admits a \emph{slice decomposition} of size $(r_1,r_2,\dots,r_d)$ if there exist subspaces $U_i\leq \CC^{n_i}$ of dimensions $r_i$ such that \[
t \in U_1 \otimes\CC^{n_2}\otimes\dots\otimes\CC^{n_d} + \dots + \CC^{n_1}\otimes\dots\otimes\CC^{n_{d-1}}\otimes U_d. \]
\end{defn}

\begin{defn}[Slice rank]
    For $\xi\in\Xi\coloneqq\{\xi\in\RR_{\geq 0}^d \mid \max_i\xi_i = 1\}$, the \emph{$\xi$-weighted slice rank} of a tensor $t \in \CC^{n_1} \ot \cdots \ot \CC^{n_d}$ is defined as
    \begin{align*}
        \SR_\xi(t) \coloneqq \min\{&r_1^{1/\xi_1}+r_2^{1/\xi_2}+\dots+r_d^{1/\xi_d}\mid t\emph{ admits a}\\&\emph{slice decomposition of size }(r_1,r_2,\dots,r_d)\}.
    \end{align*}
    If $\xi_i=0$ for some coordinate, then $r_i$ is required to be zero and the corresponding term $r_i^{1/\xi_i}$ is left out.
    The \emph{slice rank} of~$t$ is defined as $\SR(t)=\SR_{(1,\dots,1)}(t)$, i.e., $\xi=(1,\dots,1)$ is the all-ones vector.
\end{defn}

\noindent
It is clear from the definition that $\SR_\xi(g\cdot t)=\SR_\xi(t)$ for~$g\in\GL$.

In fact, by choosing base changes $g_i\in\GL(n_i)$ which maps~$U_i$ to a coordinate subspace, we see that the slice rank can be written as a minimum of hypergraph vertex cover numbers.
We first define the latter quantity and then state this observation.

\begin{defn}[Vertex cover number]
\label{def:vertex-cover}
For $\xi\in\Xi$, we define $\xi$-weighted vertex cover number $\tau_{\xi}(\HH_t)$ as the value of the optimization problem
\begin{align*}
\tau_\xi(\HH_t) = \min&\sum_{i=1}^{d} \left(\sum_{j=1}^{n_i}  u_{ij}\right)^{1/\xi_i}\\
\text{s.t. }\ & u_{1,i_1} + u_{2,i_2} + \dots + u_{d,i_d}\geq 1\\
&\qquad(\text{for every}\ e=(i_1,i_2,\dots,i_d)\in E(\HH))\\
&u \in \{0,1\}^{n_1} \times \cdots \times \{0,1\}^{n_d}.
\end{align*}
\end{defn}

\begin{lem}
For every tensor $t$ and every $\xi\in\Xi$, we have~$\SR_\xi(t) = \min_{g\in\GL} \tau_\xi(\HH_{g\cdot t})$.
\end{lem}

We will now prove that the regularization of the weighted slice rank and the weighted vertex cover number satisfy an analogous relation.
The former is known as the asymptotic slice rank and it is defined as follows:

\begin{defn}[Asymptotic slice rank]
    For $\xi\in\Xi$, the \emph{$\xi$-weighted asymptotic slice rank} of a tensor $t \in \CC^{n_1} \ot \cdots \ot \CC^{n_d}$ is defined as the limit (whenever it exists)
    \begin{equation}\label{eq:weighted-asymptotic-slice-rank}
        \widetilde{\SR}_\xi(t) \coloneqq \lim_{n\rightarrow\infty}\SR_\xi(t^{\otimes n})^{1/n}.
    \end{equation}
    The \emph{asymptotic slice rank} of~$t$ is defined as the special case~$\widetilde{\SR}(t)=\widetilde{\SR}_{(1,\dots,1)}(t)$.
\end{defn}

We note that since the slice rank is neither sub- nor super-multiplicative (\cite[Example~5.2]{CVZ2018}), it is not a priori clear why the limit should exist.

\begin{thm}[\cite{CVZ2018,CLZ-23}]\label{thm:weighted-cvz-formula}
For every tensor~$t \neq 0$ and every~$x \in \Xi$, the limit~\eqref{eq:weighted-asymptotic-slice-rank} exists and can be computed as
\[
    \widetilde{\SR}_\xi(t) = \max_{p\in\Delta(t)} \min_i 2^{H(p_i)/\xi_i} = \min_{\theta\in\Theta} F_\theta(t)^{1/\braket{\theta,\xi}},
\]
where $F_\theta$ are the quantum functionals.
In particular,
\[
    \widetilde{\SR}(t) = \max_{p\in\Delta(t)} \min_i 2^{H(p_i)} = \min_{\theta\in\Theta} F_{\theta}(t).\]
\end{thm}

We will now show that the asymptotic slice rank of a tensor can be computed as the minimum of the \emph{asymptotic vertex cover number} of $\HH_{g\cdot t}$ over $g\in\GL$.
If $\mathcal{G}=(V_1(\mathcal{G})\sqcup\dots\sqcup V_d(\mathcal{G}), E(\mathcal{G}))$ and $\HH=(V_1(\HH)\sqcup\dots\sqcup V_d(\HH), E(\HH))$ are two $d$-uniform, $d$-partite hypergraphs then their Kronecker product $\mathcal{G}\times \HH$ is defined as the hypergraph with vertex sets $\bigsqcup_{i=1}^d V_i(\mathcal{G})\times V_i(\HH)$ and edge set which consists of tuples $((v_1,v'_1),(v_2,v'_2),\dots, (v_d,v'_d))$ for~$(v_1,\dots,v_d)\in E(\mathcal{G})$ and $(v'_1,\dots,v'_d)\in E(\HH)$.
In particular, we have $E(\mathcal{G}\times \HH)= E(\mathcal{G})\times E(\HH)$ via the identification $((v_1,v'_1),(v_2,v'_2),\dots, (v_d,v'_d)) = ((v_1,\dots,v_d),(v_1',\dots,v_d'))$.
If $t$ and $s$ are two tensors, this definition has the property that \[
\HH_{t\otimes s} = \HH_t \times \HH_s.
\]

\begin{defn}[Asymptotic vertex cover number]
For $\xi\in\Xi$, the \emph{$\xi$-weighted asymptotic vertex cover number} of a hypergraph~$\HH$ is defined as (if the limit exists)
\begin{equation}\label{eq:weighted-asymptotic-vertex-cover-number}
\tilde{\tau}_\xi(\HH)\coloneqq \lim_{n\rightarrow\infty}\tau_\xi(H^{\times n})^{1/n}.
\end{equation}
The \emph{asymptotic vertex cover number} of~$\HH$ is defined as the special case $\tilde{\tau}(\HH)=\tilde{\tau}(\HH)_{(1,\dots,1)}$.
\end{defn}

Tao and Sawin \cite{Tao-Sawin-16} (see also \cite[Proposition~4.3.3]{NvdB}) proved a formula that shows that the asymptotic vertex cover number can be computed by maximizing the minimum entropy of the marginal distributions of probability distributions on~$E(\HH_t)$.
Equivalently, this can be stated as an entropy optimization over the support polytope~$\Omega(t)$.
We provide a weighted version of their result and state an alternative description.

\begin{prop}
\label{prop:tao-sawin-weighted}
For every tensor~$t \neq 0$ and every~$\xi \in \Xi$, the limit~\eqref{eq:weighted-asymptotic-vertex-cover-number} exists and we have
\begin{align*}
    \tilde{\tau}_{\xi}(\HH_t)
    &= \max_{p\in\Omega(t)} \min_i 2^{H(p_i)/\xi_i}\\
    &= \min_{\theta\in\Theta} \left( \max_{p \in \Omega(t)} 2^{\sum_{i=1}^d \theta_i H(p_i)} \right)^{1/\braket{\theta,\xi}}.
\end{align*}
\end{prop}
\begin{proof}
    The second formula follows from the first by Neumann's minimax theorem for quasiconvex-quasiconcave functions, see~\cite[Thm.~50]{CLZ-23}.
    We now prove the first formula by a straightforward generalization of the proof for $\xi=(1,\dots,1)$ by Tao and Sawin~\cite{Tao-Sawin-16}.
    Note that a vertex cover can be viewed as a partition \[
    E(\HH) = \Gamma_1 \cup \Gamma_2 \cup \dots\cup \Gamma_d,
    \] with $\Gamma_i$ being the edges covered by the $i$-th vertex set of the cover.
    In particular, we have \[
    \tau(\HH) = \min_{E(\HH) = \bigcup_{i=1}^d \Gamma_i} \sum_i |\pi_i(\Gamma_i)|,
    \] where $\pi_i:E(\HH)\rightarrow V_i(\HH)$ denotes the $i$-th projection map.

    Tao and Sawin proved that for every probability distribution~$P$ on $E(\HH)$ and for any vertex cover $E(\HH^{\times n})=E(\HH)^{\times n}=\Gamma_{1,n}\cup\dots\cup\Gamma_{d,n}$, there exists an index $j\in [d]$ such that \begin{equation}
    \label{eq:sawin-tao-lb}
    |\pi_j(\Gamma_{j,n})| \geq 2^{(H(p_j) + o(1))n},
    \end{equation}
    where $p_j\; (j=1,2,\dots,d)$ denote the marginal distributions of $P$.
    Then, by
    \cref{eq:sawin-tao-lb} we have $\sum_{i=1}^d |\pi_i(\Gamma_{i,n})|\geq \min_j |\pi_j(\Gamma_{j,n})| \geq \min_j 2^{ (H(p_j)+o(1))n}$, which holds for every vertex cover and every distribution $P$.
    For $\xi\in\Xi$, \eqref{eq:sawin-tao-lb} implies~$|\pi_j(\Gamma_{j,n})|^{1/\xi_j} \geq 2^{n (H(p_j)/\xi_j+o(1))}$, and~$\sum_{i=1}^d |\pi_i(\Gamma_{i,n})|^{1/\xi_i}\geq \min_j 2^{(H(p_j)/\xi_j+o(1))n}$.
    Taking $n\rightarrow\infty$ and maximizing over the probability distributions, we get
    \begin{equation*}
    \lim_{n\rightarrow\infty}\tau_\xi(\HH_t^{\times n})^{1/n} \geq \max_{p\in\Omega(t)} \min_i 2^{H(p_i)/\xi_i}.
    \end{equation*}

    Conversely, Tao and Sawin proved that there exist~$p^{(1)},\dots,p^{(l)}\in\Omega(t)$ with $l=2^{o(n)}$ such that there exists a covering \[
    E(\HH_t^{\times n}) = \bigcup_{i=1}^l \Gamma_n^{(i)}
    \] with the property that for every $i\in [l]$ and $j\in [d]$ we have \[
    |\pi_j(\Gamma_n^{(i)})| \leq 2^{\left( H(p^{(i)}_j) + o(1) \right)n}.
    \]
    Let's denote by $j_i\in [d]$ the index that satisfies
\[
\min_{j\in [d]} H(p^{(i)}_j)/\xi_j = H(p^{(i)}_{j_i})/\xi_{j_i},
\] and regroup as $\Gamma_{k,n}\coloneqq \bigcup_{j_i=k} \Gamma_n^{(i)}$.
    Since $l=2^{o(n)}$, we have
    \begin{align*}
        \tau_{\xi}(\HH_t^{\times n})&\le \sum_{k = 1}^d|\pi_k(\Gamma_{k, n})|^{1/\xi_k} \le \sum_{i = 1}^l\min_{j \in [d]}2^{(H(p_j^{(i)})/\xi_j + o(1))n}\\ &\le 2^{o(n)}\max_{p \in \Omega(t)}\min_{j \in [d]}2^{(H(p_j^{(i)})/\xi_j + o(1))n}.
    \end{align*}
    By taking $n\rightarrow\infty$ we get
    \[
        \lim_{n \to \infty}\tau_{\xi}(\HH_t^{\times n})^{1/n} \leq \max_{p \in \Omega(t)}\min_j2^{H(p_j)/\xi_j}. \qedhere
    \]
\end{proof}

We can now prove \cref{thm:intro-weighted-sr}, which generalizes \cref{cor:slice-rank} stated in the introduction:

\begin{proof}[Proof of \cref{thm:intro-weighted-sr}]
    This follows at once from \cref{thm:weighted-cvz-formula,thm:support-equals-quantum,prop:tao-sawin-weighted}:
    \begin{align*}
        \widetilde{\SR}_\xi(t)
    &= \min_{\theta\in\Theta} F_\theta(t)^{1/\braket{\theta,\xi}}
    = \min_{\theta\in\Theta} \left( \zeta^\theta(t) \right)^{1/\braket{\theta,\xi}}\\
    &= \min_{\theta\in\Theta} \min_{g \in \GL} \left( \max_{p \in \Omega(g \cdot t)} 2^{\sum_{i=1}^d \theta_i H(p_i)} \right)^{1/\braket{\theta,\xi}}\\
    &= \min_{g \in \GL} \tilde{\tau}_{\xi}(\HH_t).
    \qedhere
    \end{align*}
\end{proof}

\begin{rem}
Alternatively, \cref{thm:intro-weighted-sr} can be proved by applying \cref{thm:moment-polytope-general-q} to the function $F_\xi(p) \coloneq -\min_i H(p_i)/\xi_i$ (with $H(p_i)/\xi_i=\infty$ if $\xi_i=0$), which is a convex symmetric l.s.c.\ function on $\RR^{n_1} \times \cdots \times\RR^{n_d}$ that takes a finite value at any point~$p=(p_1, \ldots, p_d) \in \RR^{n_1} \times \cdots \times\RR^{n_d}$ such that each $p_i$ is a probability distribution, hence in particular on any moment and support polytope.
\end{rem}

\subsection{\texorpdfstring{$G$}{G}-stable rank}\label{sub:G stable}
Next we discuss a tensor parameter introduced by Derksen~\cite{Derksen-22} called the $G$-stable rank.
Intuitively, it measures how fast $t$ goes to zero under the action of one-parameter subgroups of~$\GL$.
Derksen gave a general definition that applies to any perfect field and proved the following formula~\cite[Cor.~4.2 and Thm.~4.4]{Derksen-22}, which we will take as the definition:

\begin{defn}[$G$-stable rank, fractional vertex cover]
\label{def:G-rk-fractional}
For~$\alpha\in\RR_{>0}^d$, we define the \emph{$G$-stable $\alpha$-rank} of a tensor~$t \in \CC^{n_1} \ot \cdots \ot \CC^{n_d}$ as
\begin{align}
\label{eq:g-stable-rank}
    \rk^G_\alpha(t) \coloneqq \min_{g \in \GL} \tau_{\alpha}^{f}(\HH_{g\cdot t}),
\end{align}
where $\HH_{g \cdot t}$ denotes the hypergraph associated to the tensor~$g \cdot t$ (as in \cref{sub:asymptotic slice rank}) and $\tau_{\alpha}^f$ denotes the $\alpha$-weighted \emph{fractional vertex cover} of a hypergraph, defined by the following linear program:
 \begin{equation}
\label{eq:fractional-vertex-cover}
\begin{split}
\tau^f_\alpha(\HH) = \min \; & \sum_{i=1}^{d} \sum_{j=1}^{n_i} \alpha_i u_{ij}\\
\text{s.t. }\, & u_{1,i_1} + u_{2,i_2} + \dots + u_{d,i_d}\geq 1\\
&\quad(\text{for every }e=(i_1,i_2,\dots,i_d)\in E(\HH))\\
&u \in \RR^{n_1 + \dots + n_d}_{\geq 0}.
\end{split}
\end{equation}
The \emph{$G$-stable rank} is then defined as~$\rk^G(t) \coloneqq \rk^G_{(1,\dots,1)}$, i.e., $\alpha=(1,\dots,1)$ is the all-ones vector.
\end{defn}

The formula in \cref{eq:g-stable-rank} should be interpreted as a formula in the spirit of Strassen's support functionals.
Indeed, it is not hard to see that $\tau^f_\alpha(\HH_t)$ only depends on the tensor's support polytope.
The following proposition gives a concise expression:

\begin{prop}
\label{prop:fractional-cover-reciprocal}
For every tensor~$t \neq 0$ and every~$\alpha\in\RR_{>0}^d$, we have
\[
    \tau_{\alpha}^f(\HH_t) = \max_{p\in\Omega(t)}\min_{i\in [d]} \frac{\alpha_i}{\norm{p_i}_\infty}.
\]
\end{prop}
\begin{proof}
By LP duality, the fractional vertex cover number of a hypergraph equals its fractional matching number.
In other words:
\begin{equation*}
\begin{split}
\tau^f_\alpha(\HH) = \max \; & \sum_{e\in E(\HH)} y_e\\
\text{s.t. }\, & \sum_{e \in E(\HH), \ e_i = j} y_e\leq \alpha_i\quad (\forall i \in [d], j \in [n_i])\\
&y\geq 0.
\end{split}
\end{equation*}
By a change of variables $y = c z$, we may therefore rewrite~$\tau^f_\alpha(\HH_t)$ as  the value of the following optimization problem:
\[
\begin{split}
\tau^f_\alpha(\HH_t) = \max  &\quad c\\
\text{s.t. }\, & \sum_{e\in E(\HH)} z_e = 1,\\
&\sum_{e \in E(\HH), \ e_i = j} z_e\leq \frac{\alpha_i}{c}\quad (\forall i \in [d], j \in [n_i]),\\
&z\geq 0, c \geq 0.
\end{split}
\]
Note that $z$ is a probability distribution on the edge set of $\HH_t$ (i.e., the support of $t$), so its marginals \[
(p_i)_j \coloneqq \sum_{e_i=j} z_e
\] satisfy $(p_1,\dots,p_d)\in \Omega(t)$.
Conversely, every $p\in\Omega(t)$ is the marginal distributions of a probability distribution on $E(\HH_t)$.
    This implies that \[
    \tau_\alpha^f(\HH_t) = \max_{\substack{p\in\Omega(t)\\ \forall i, \norm{p_i}_\infty\leq \frac{\alpha_i}{c}}} c = \max_{p\in\Omega(t)}\min_i \frac{\alpha_i}{\norm{p_i}_\infty}.
    \qedhere
    \]
\end{proof}

By an application of \cref{thm:moment-polytope-general-q}, we will now show that~$\rk^G_{\alpha}(t)$ also admits a formula as an optimization problem over the entanglement polytope $\Delta(t)$, reproducing a result by Derksen.

\begin{cor}[{\cite[Thm.~5.2]{Derksen-22}}]
For every tensor~$t \neq 0$ and every~$\alpha\in\RR_{>0}^d$, we have
\begin{align*}
    \rk^G_{\alpha}(t) = \max_{p\in\Delta(t)}\min_{i\in [d]} \frac{\alpha_i}{\norm{p_i}_\infty}.
\end{align*}
\end{cor}
\begin{proof}
By \cref{prop:fractional-cover-reciprocal} and \eqref{eq:g-stable-rank} we get
\[
\frac{1}{\rk^G_{\alpha}(t)} = \max_{g\in \GL}\min_{p\in \Omega(g\cdot t)} \max_{i\in [d]}\frac{\norm{p_i}_{\infty}}{\alpha_i}.
\]
The function $F_\alpha(p) \coloneqq \max_{i\in [d]} \frac{1}{\alpha_i}\norm{p_i}_{\infty}$ is a convex symmetric continuous function on $\RR^{n_1}\times\cdots\times\RR^{n_d}$.
Hence, by an application of \cref{thm:moment-polytope-general-q}, we obtain the desired result:
\begin{equation*}
    \frac{1}{\rk^G_{\alpha}(t)} = \min_{p\in\Delta(t)} \max_{i\in [d]}\frac{\norm{p_i}_\infty}{\alpha_i},
    \qedhere
\end{equation*}
\end{proof}

\subsection{Non-commutative rank}\label{sub:ncrank}
A 3-tensor $t\in\CC^{n}\otimes\CC^n\otimes\CC^m$ can be viewed as a matrix tuple $t=(A_1,A_2,\dots,A_m)$, where the matrices~$A_i\in\CC^{n\times n}$ are its slices in the third direction.
For simplicity of presentation, we will assume that the slices satisfy
\begin{equation}
\label{eq:no-isolated-vertex}
\bigcap_{i}\ker(A_i) = 0 \quad\text{or}\quad \bigcap_i \ker(A_i^\dag)=0.
\end{equation}

\begin{defn}[Noncommutative rank]
The \emph{non-commutative rank} $\ncrk(t)$ of a 3-tensor or matrix tuple~$t=(A_1,A_2,\dots,A_m)$ is defined as the rank of the matrix of indeterminates
\[
    A(x) \coloneqq x_1 A_1 + x_2 A_2 +\dots + x_m A_m,\quad \ncrk(t) \coloneqq \rk(A(x)),
\]
where $\rk(A(x))$ is the matrix rank over the \emph{free skew field}~$\CC\ (\!\!\!\!<\!\!x_1,x_2,\dots,x_m\!\!>\!\!\!\!)$.
\end{defn}

In \cite{Hirai_ncrank}, Hirai used a version of the strong duality theorem to recast the non-commutative rank as an optimization problem over the moment polytope with respect to the action of~$\GL(n)\times\GL(n)$ on matrix tuples $\CC^{n}\otimes\CC^n\otimes\CC^m$ given by $(g,h) \cdot t := (g \ot h \ot I) t$.
This action is often called the \emph{left-right action}, as it sends a matrix tuple $(A_1,A_2,\dots,A_m)$ to $(g A_1 h^T,g A_2 h^T,\dots,g A_m h^T)$.
Then Hirai's result can be stated as follows:

\begin{thm}[{\cite[Thm.~1.4]{Hirai_ncrank}}]\label{thm:hiroshi ncrank}
For every 3-tensor $0 \neq t \in \CC^n \ot \CC^n \ot \CC^m$ whose slices satisfy \eqref{eq:no-isolated-vertex}, we have
\begin{align*}
    \ncrk(t)
    &= n - \frac{n}{2}\min_{p \in \Delta_{\mathrm{LR}}(t)} \norm*{p - \frac{\bm{1}_{2n}}{n}}_1\\
    &= n - \frac{n}{2}\min_{(p_1,p_2)\in \Delta_{\mathrm{LR}}(t)} \norm*{p_1 - \frac{\bm{1}_n}{n}}_1 + \norm*{p_2 - \frac{\bm{1}_n}{n}}_1,
\end{align*}
where we write $\Delta_{\mathrm{LR}}(t)$ for the moment polytope of~$t$ associated with the left-right action and~$\bm{1}_k \in \RR^k$ for the all-ones vectors.
\end{thm}

We will show how to reproduce this result using \cref{thm:moment-polytope-general-q}.
Our starting point is the following well-known formula for the non-commutative rank established by Fortin and Reutenauer~\cite{Fortin-Rautenauer-04}:
for any matrix tuple $t=(A_1,\dots,A_m)$,
\begin{gather*}
\ncrk(t) = 2n - \max_{(g,h)\in\GL(n)\times\GL(n)}\, (a+b),\\
\text{where }\ gA_i h^\top = \left(
\begin{array}{c|c}
* & * \\ \hline
* & 0_{a\times b}
\end{array}
\right)\quad(\forall i\in [m]).
\end{gather*}
The combinatorial quantity $a+b$ can be understood via a bipartite graph $\BB_t$ that is naturally associated with the matrix tuple, in the following way.
Let $\BB_t$ have left and right vertex sets both indexed by $[n]$, and include an edge $(i,j)\in E(\BB_t)$ whenever there exists $k\in [m]$ such that $(A_k)_{ij}\neq 0$.
Equivalently, $(i,j)\in E(\BB_t)$ if $(e_i,e_j,e_k)\in\supp(t)$ for some~$k$.
In graph-theoretic language, $2n - (a+b)$ is the size of a minimum vertex cover of the graph~$\BB_t$, i.e., the vertex cover number~$\tau(\BB_t)$.
Hence, the Fortin-Reutenauer formula states that
\begin{equation}\label{eq:fortin-reutenauer}
\ncrk(t) = \min_{(g,h)\in\GL(n)\times\GL(n)} \tau(\BB_{(g\otimes h\otimes I_m) t}).
\end{equation}
K\H{o}nig's theorem implies that the vertex cover number of a bipartite graph~$\BB$ equals the \emph{fractional} vertex cover number.
That is:
\begin{equation}
\label{eq:vertex-cover}
\tag{$\star$}
\begin{split}
\tau(\BB) = \min\; &\sum_{i=1}^n u_i + \sum_{j=1}^n v_j \\
\text{s.t. } &\forall (i,j)\in E(B),\; u_i+v_j\geq 1\\
\quad & u,v\geq 0.
\end{split}
\end{equation}

The next proposition shows that the value of the above linear program for~$\BB_t$ can be obtained by a convex optimization problem over the polytope~$\Omega_{\mathrm{LR}}(t) \coloneqq \conv\{(e_i,e_j)\mid (i,j)\in E(\BB_t)\}$, which is the support polytope for the left-right action.

\begin{prop}
\label{prop:l1-vertex-cover}
Suppose $0\neq t\in\CC^n\otimes\CC^n\otimes\CC^m$ and \eqref{eq:no-isolated-vertex} holds.
    Then,
    \[
    \tau(\BB_t) = n - \frac{n}{2}\min_{p\in \Omega_{\mathrm{LR}}(t)} \,\norm*{p-\frac{\bm{1}_{2n}}{n}}_1,
    \]
    where $\bm{1}_{2n}\in\RR^{2n}$ denotes the all-ones vector.
\end{prop}
\begin{proof}
Note that the assumption that \eqref{eq:no-isolated-vertex} holds translates to that at least either of the vertices sets of $\BB_t$ has no isolated vertices.

Let's denote $\omega_{2n}\coloneqq \frac{\bm{1}_{2n}}{n}$ and
consider $F(p)\coloneqq \norm{p - \omega_{2n}}_1$, defined on $\RR^{2n}$.
Using the fact that $\norm{p}_1 = \sup_{\norm{x}_\infty\leq 1} p^T x$, we calculate the Legendre--Fenchel conjugate \[
\begin{split}
    F^*(x) &= \sup_{p\in\RR^{2n}}\, \langle p,x\rangle - \norm{p-\omega_{2n}}_1 \\
    &= \sup_{p\in\RR^{2n}} \, \left(\langle p,x\rangle + \min_{\norm{y}_{\infty}\leq 1} -\langle p-\omega_{2n},y\rangle \right)\\
    &= \min_{\norm{y}_{\infty}\leq 1} \sup_{p\in \RR^{2n}} \langle p, x-y\rangle +\langle \omega_{2n},y\rangle\\
    &= \begin{cases}
        \langle \omega_{2n},x\rangle & \text{ if } \norm{x}_{\infty}\leq 1\\
        \infty & \text{ otherwise }
    \end{cases}
\end{split}
\] where we use von Neumann's minimax theorem in the third equality.

Now consider the indicator function $\delta_{\Omega_{\mathrm{LR}}(t)}$ of $\Omega_{\mathrm{LR}}(t)$, defined as $\delta_{\Omega_{\mathrm{LR}}(t)}(p)=0$ if $p\in \Omega_{\mathrm{LR}}(t)$ and $\infty$ otherwise.
Its conjugate \[
\delta^*_{\Omega_{\mathrm{LR}}(t)}(x) = \max_{p\in \Omega_{\mathrm{LR}}(t)} \langle p, x\rangle \stackrel{x = (a,b)}{=} \max_{(i,j)\in E(\BB_t)} a_i + b_j
\]
is the support function of $\Omega_{\mathrm{LR}}(t)$.

We have
\begin{align*}
    &\min_{p\in \Omega_{\mathrm{LR}}(t)} F(p)
    = \sup_{x\in\RR^{2n}} -F^{*}(x) - \delta_{\Omega_{\mathrm{LR}}(t)}^*(-x)\\
    &= \max_{\norm{(a,b)}_{\infty}\leq 1}\left[ -\frac{1}{n}\left(\sum_i a_i + \sum_j b_j \right)+\min_{(i,j)\in E(B)} a_i+b_j\right],
\end{align*}
where we use $\min_{p\in \Omega_{\mathrm{LR}}(t)}F(p)=\min_{p\in\RR^{2n}} F(p)+\delta_{\Omega_{\mathrm{LR}}}(p)$ and apply Fenchel's duality~\cite[Theorem~31.1]{Rockafellar1970} in the first equality.

By a change of variables $(c,d)=-(a,b)$, we can write the above maximization problem as a linear program (note that~$\norm{(a,b)}_{\infty}\leq 1$ is equivalent to $-\bm{1}_n\leq a,b\leq\bm{1}_n$) :\begin{equation}
\label{eq:LP}
\begin{split}
\max\;\; &\frac{1}{n}\left(\sum_{i}c_i +\sum_j d_j\right) - t \\
\text{s.t. } \;\;&\;\; c_i+d_j \leq t, \quad \text{for every }(i,j)\in E(\BB_t), \;\\
\;\; & -\bm{1}_n \leq c,d\leq \bm{1}_n.
\end{split}
\end{equation}
It remains to show that $n-\frac{n}{2}\times\text{(the optimal value of \eqref{eq:LP})}$ equals the optimal value of \eqref{eq:vertex-cover}.

We now prove that \eqref{eq:LP} has an optimal solution with~$t=0$:
Suppose $x=(c,d,t)$ is an optimal solution.
By assumption~\eqref{eq:no-isolated-vertex} and symmetry, we assume that the left vertex set (which corresponds to $c_i$'s) has no isolated vertices, and set~$\tilde{c}_i = c_i - t$.
We claim that $\tilde{x}=(\tilde{c},d,0)$ is also an optimal solution.
It is easy to check that the the objective function takes the same value on $\tilde{x}$.
Hence, we only need to show that $\tilde{x}$ is feasible.
For $(i,j)\in E(\BB_t)$, we have~$\tilde{c}_i+d_j = c_i + d_j - t \leq 0$, so it remains to show that $-\bm{1}_n\leq \tilde{c} \leq \bm{1}_n$.
There are two cases.
If $t<0$, then for all $i$, $\tilde{c}_i>c_i$.
Suppose there exists~$i$ such that~$\tilde{c}_i>1$.
Since $\BB_t$ has no isolated vertices on the left by \eqref{eq:no-isolated-vertex}, there is a $j$ such that $(i,j)\in E(\BB_t)$.
Then,~$\tilde{c}_i > 1$ implies $c_i > 1+t$.
Since~$d_j\geq -1$, $c_i+d_j > t$, this is a contradiction.
If $t>0$, then $\tilde{c}_i<c_i$.
If there exists~$i$ with $\tilde{c}_i<-1$, then $c_i < -1 + t$.
Now, for every~$j$ with~$(i,j)\in E(\BB_t)$, we have $c_i+d_j < t$.
Note that we necessarily have $\max c_i+d_j=t$, since otherwise we can decrease $t$.
By $c,d\leq 1$, we deduce that $t\leq 2$ so $c_i<-1+t$ implies that $c_i<1$.
This contradicts the assumption that $x$ is an optimal solution, since we can increase $c_i$ by a small value~$\varepsilon>0$, which in turn increases the objective value, without violating any constraint.

We now set $t=0$ in \eqref{eq:LP}.
By a change of variables~$u_i \coloneqq \frac{c_i-1}{2}, v_i \coloneqq \frac{d_i-1}{2}$, the value of \eqref{eq:LP} equals the value of the linear program
\begin{equation}
\label{eq:LP2}
\begin{split}
\max\;\; & 2 - \frac{2}{n}\left(\sum_i u_i + \sum_j v_j \right) \\
\text{s.t. }\; & u_i + v_j \geq 1 \quad (\text{for every } (i,j)\in E(\BB_t))\\
& u,v\geq 0.
\end{split}
\end{equation}
Hence \[
\min_{p\in \Omega_{\mathrm{LR}}(t)} \norm{p-\omega_{2n}}_1 = 2 - \frac{2}{n}\min_{\substack{u,v\geq 0 \\ \forall (i,j)\in E(\BB_t), u_i+v_j\geq 1}} u_i + v_j
\]
and this finishes the proof.
\end{proof}

\begin{proof}[Proof of \cref{thm:hiroshi ncrank}]
The Fortin-Reutenauer formula~\eqref{eq:fortin-reutenauer}, along with \cref{prop:l1-vertex-cover}, implies that
\begin{align*}
&\ncrk(t) =\\
&n - \frac{n}{2}\max_{(g,h)\in \GL(n_1)\times\GL(n_2)}\min_{(p_1,p_2)\in \Omega_{\mathrm{LR}}((g\otimes h\otimes I_m) t)} \, F(p_1,p_2),
\end{align*}
where $F(p)\coloneqq \norm{p-\frac{\bm{1}_{2n}}{n}}_1= \norm{p_1 - \frac{\bm{1}_n}{n}}_1 + \norm{p_2 - \frac{\bm{1}_n}{n}}_1$ for~$p=(p_1,p_2) \in \RR^{n_1} \times \RR^{n_2}$.
Since $F$ is convex, symmetric and continuous, the claim follows at once from \cref{thm:moment-polytope-general-q} (with $\pi$ the left-right action).
\end{proof}

\section*{Acknowledgements}
The authors thank Hiroshi Hirai, Jeroen Zuiddam, and Maxim van den Berg for valuable discussions, and the anonymous reviewers for valuable comments.
Part of this work was conducted while MW was visiting Q-FARM and the Leinweber Institute for Theoretical Physics at Stanford University and the Simons Institute for the Theory of Computing at UC Berkeley.

\bibliographystyle{alphaurl}
\bibliography{references}

\end{document}